\documentclass[11pt]{article}

\usepackage{fullpage}
\usepackage[dvipsnames,usenames]{xcolor}
\usepackage{subfig}
\usepackage{amsthm,amssymb,amsmath}  
\usepackage{dsfont}
\usepackage{xspace,enumerate}
\usepackage[utf8]{inputenc}
\usepackage{thmtools}
\usepackage{thm-restate}
\usepackage{blkarray}
\usepackage{chngcntr}
\usepackage{authblk}
\usepackage[hypertexnames=false,colorlinks=true,urlcolor=Blue,citecolor=Green,linkcolor=BrickRed]{hyperref}

\usepackage[ruled]{algorithm}
\usepackage[noend]{algpseudocode}
\algnewcommand{\LineComment}[1]{\State \(\triangleright\) #1}
\usepackage{graphicx}
\usepackage[noadjust]{cite}
  \usepackage{microtype}
    \usepackage{amsfonts}
    \usepackage{wasysym}
  \usepackage{makecell}
  \usepackage{tabularx}
  \usepackage{comment}
  \usepackage{todonotes}
  \usepackage[capitalise]{cleveref}
  \usepackage{verbatim}
\usepackage{units}
\usepackage{color}
\definecolor{darkblue}{rgb}{0,0.08,0.45}
  \usepackage{epsfig}
    \usepackage{graphics}
  \usepackage{graphicx}
\usepackage{marvosym}
\usepackage{booktabs}
\usepackage{etoolbox}
\usepackage{enumitem}

\newcommand{\AG}{\mathcal{A}}

\newcommand{\ignore}[1]{}

\newcommand{\e}{\mathbb{E}}
\newcommand{\p}{\mathbb{P}}

\def\eps{\varepsilon}

\theoremstyle{plain}
\newtheorem{theorem}{Theorem}
\newtheorem{lemma}[theorem]{Lemma}

\newtheorem{fact}[theorem]{Fact}

\newtheorem{definition}[theorem]{Definition}

\newtheorem{claim}[theorem]{Claim} 
\newtheorem*{conjecture*}{Conjecture}

\newtheorem{question}{Question}

\author[1]{Amir Abboud\footnote{Supported by an Alon scholarship and a research grant from the Center for New Scientists at the Weizmann Institute of Science.}}
\author[1]{Nathan Wallheimer}

\affil[1]{
Weizmann Institute of Science\\
\href{mailto:amir.abboud@weizmann.ac.il}{amir.abboud@weizmann.ac.il}, \href{mailto:nathanwallh@gmail.com}{nathanwallh@gmail.com}}

\title{Worst-Case to Expander-Case Reductions}

\begin{document}

\maketitle

\begin{abstract}
In recent years, the \emph{expander decomposition method} was used to develop many graph algorithms, resulting in major improvements to longstanding complexity barriers.
This powerful hammer has led the community to (1) believe that most problems are as easy on worst-case graphs as they are on expanders, and (2) suspect that expander decompositions are the key to breaking the remaining longstanding barriers in fine-grained complexity.

We set out to investigate the extent to which these two things are true (and for which problems).
Towards this end, we put forth the concept of \emph{worst-case to expander-case} self-reductions.
We design a collection of such reductions for fundamental graph problems, verifying belief (1) for them. The list includes $k$-Clique, $4$-Cycle, Maximum Cardinality Matching, Vertex-Cover, and Minimum Dominating Set.
Interestingly, for most (but not all) of these problems the proof is via a simple gadget reduction, \emph{not} via expander decompositions, showing that this hammer is effectively useless against the problem and contradicting (2).

\end{abstract}

\thispagestyle{empty}
\newpage
\setcounter{page}{1}

\section{Introduction}

One of the most effective techniques in modern graph algorithms has been the \emph{expander decomposition method}. 
To solve a difficult problem on a worst-case instance $G$ on $m$ edges, it roughly proceeds as follows.
\begin{enumerate}
\item Decompose $G$ into the disjoint union of \emph{expanders} $H_1,\ldots,H_k$ (also called \emph{clusters}) plus $o(m)$ \emph{outer edges}; such an expander decomposition can be computed in $m^{1+o(1)}$ time \cite{KVV04,OV11,OSV12,ST13,SW19,chuzhoy2019deterministic,LS21,Kaplan22}. 
\item Run some computation on each expander $H_i$, utilizing the special properties of expanders to gain a speed-up, e.g., because the search space is restricted on expanders or because some processes can be accelerated.
\item Run some computation on the $o(m)$ outer edges (e.g., recurse), utilizing the fact that there are not too many such edges.

\end{enumerate}
The last two steps are problem-specific and often involve ingenious techniques. To keep this work focused, we choose to restrict ourselves to the following definition of expanders (that is most popular in our context).
    \begin{definition}[Expander]
    \label{def:expander}
    A graph $G=(V,E)$ is called an $\phi$-\emph{expander} if its \emph{conductance} is:
    $$
    \phi(G):= \min_{S \subset V} 
    \frac{|E(S,V\setminus S)|}{\min(vol(S), vol(V \setminus S))} \geq \phi 
    $$
    where $vol(S)$ is the sum of degrees of all nodes in $S$.\footnote{If $G$ is a singleton then $\phi(G) := 1$.}
    \end{definition}
The expander decomposition method itself is common to several celebrated results breaking longstanding bounds for classical problems, e.g., \cite{ST14,KLOS14,NSW17,CK19,AKT21,YS22}, leading the community to believe the intuition that ``all graph problems are as easy on worst-case graphs as they are on expanders (up to sub-polynomial factors)''. In other words, if we can solve the problem fast on expanders then we can also solve it fast on worst-case graphs. Moreover, the expander-decomposition method is viewed as the hammer proving this intuition. Of course, it's unlikely that this intuition is true for \emph{all} problems; since there are problems for which the state of the art on expanders is better than it is on worst-case graphs.\footnote{One interesting example is the problem of computing a Gomory-Hu tree of a graph. In weighted graphs, the expander case can be solved in $m^{1+o(1)}$ time (follows from Theorem~1.4 in \cite{AKT22friendly}) but the worst-case bound is $\tilde{O}(n^2)$ \cite{AKLPST21}.}

In this work, we set out to investigate the extent to which this intuition is true. 
Namely, we are looking to answer the following motivating question:
\begin{question}
\label{q:q1}
Which problems are as easy on worst-case graphs as they are on expanders? 
\end{question}
\noindent
To this end, we formalize a new notion of \emph{worst-case to expander-case self-reductions}. If a problem admits such a reduction then the answer to this question is positive, meaning that the worst-case graphs of these problems \emph{are} expanders. Essentially, the expander decomposition method can be viewed as a special case of such reductions, assuming that its three steps can be computed in significantly faster time than the best upper bound for the problem. 

Since the question is irrelevant for easy problems that are already known to be solvable in near-linear time (that is, in $\tilde{O}(m)$ time where $\tilde{O}(\cdot)$ hides poly-logarithmic factors) in the worst-case, we are looking for difficult problems, for which the best known algorithms still run significantly slower than near-linear (say, in $\Omega(n^2)$ time). To find such problems we look at the state of the art in fine-grained complexity.

Fine-grained complexity, with its reductions-based approach, has been effective at mapping out the landscape of fundamental problems that are difficult. In the context of graph problems, there is a large number of tight conditional lower bounds that are based on the conjectured hardness of detecting a triangle in a graph (Triangle Detection) and the related Minimum-Weight Triangle, $k$-Clique, Shortest Cycle, etc. (see e.g., \cite{WW18,DBLP:conf/focs/AbboudW14,AGV15,ABW18,AR18,LWW18} and the survey \cite{williams2018some}). It is known that in order to make progress on any of a large number of problems, one must first come up with a groundbreaking algorithm for Triangle Detection (and its variants). 

Notably, most reductions from Triangle Detection to other problems, e.g. Dynamic Maximum Matching, are \emph{expansion preserving} in the sense that if the Triangle Detection instance is an expander, then so is the Dynamic Maximum Matching instance. The reason is that many of these reductions employ only a small number of local changes to the graph that do not reduce the expansion by much. Consequently, if the answer to the question for Triangle Detection is positive, meaning that the hardest instances are expanders, then it is also positive for many other Triangle-hard problems. 
Thus, it is most interesting to investigate Question~\ref{q:q1} for Triangle Detection and other problems without conditional lower bounds.

But if the prevalent intuition is indeed true (due to expander decompositions), and the answer is positive for Triangle Detection, then all we have to do in order to refute the ``Triangle conjecture'' is to come up with a faster algorithm on expanders.
This may appear much easier. After all, expanders are similar in many ways to random graphs and we can solve Triangle Detection on random $G(n,p)$ graphs in expected subquadratic time.\footnote{Basic probability shows that with high probability, the graph contains at least $ 0.1 n^3 p^3$ triangles and that the maximum degree is bounded from above by $10np \log n$.
Thus, we can sample triplets of vertices and check if they induce a triangle, while alternately running the $\Delta^2 n$ algorithm for graphs with bounded degree $\Delta$ (that checks for an edge between every pair of neighbors for every vertex). This algorithm terminates in expected $\tilde{O}(\min\{1/p^3,n^3p^2\})$ time and it balances to $\tilde{O}(n^{9/5})$ when $p = n^{-3/5}$.}

Moreover, three previous results on Triangle Detection exploit its easiness on pseudo-random graphs: a combinatorial mildly subcubic algorithm using the Szemerédi Regularity Lemma \cite{BW12}, fine-grained reductions that make the graph random-like by removing dense pieces to prove hardness for approximate distance oracles \cite{ABKZ22}, and an almost-optimal \emph{distributed} algorithm using expander decompositions \cite{CPS21}. Indeed, many researchers have been wondering whether the ``hammer'' of expander decompositions is the right tool to refute these conjectures.

\begin{question}
\label{q:q2}
Are expander decompositions the key towards resolving the open questions of fine-grained complexity?
\end{question}

Despite the major advancements in fine-grained complexity over the last decade, there are still many remaining gaps in the complexity of classical problems, for which achieving tight lower bounds has been notoriously difficult. For example, we have the problems of Maximum Cardinality Matching, $4$-Cycle Detection, and All-Pairs Max-Flow. For the latter problem, a cubic upper bound by Gomory and Hu from 1961 stood for almost 60 years until it was broken (for simple graphs) using the \emph{expander decomposition method} \cite{AKT21}. This makes us even more motivated to study Question~\ref{q:q2}.

Finally, we note that both motivating questions are relevant even for NP-hard problems, where it is desirable to optimize the \emph{exact} exponential time complexity. So far, the reductions-based approach of fine-grained complexity has only found limited success in this regime (see \cite{Cygan+16} and an interesting recent barrier \cite{KM22}), leaving a state of the art that is full of gaps (see \cite{woeginger2003exact}). It is natural to wonder if expander decompositions can play a role in this regime too.

\subsection{Our Contribution: Worst-Case to Expander-Case Self-Reductions}
\label{sec:contributions}

A \emph{worst-case to expander-case} self-reduction (WTER) is a fine-grained self-reduction that translates some graph problem $\mathcal{A}$ to the same problem on $\phi$-expanders, for some value $\phi$. 
Such reductions effectively answer Question \ref{q:q1}, and can indeed be achieved for many problems via the expander decomposition method.
However, if we take a step back, forget all the technology of expander decompositions, and simply ask if a problem admits a WTER: it is natural to try a more na\"ive gadget reduction that turns any graph into an expander.


\paragraph{Direct-WTERs}
The straightforward approach to design a WTER for some problem $\AG$ is by showing an efficient transformation that takes an instance graph $G$ and transforms it into a $\phi$-expander $G'$, such that the solution $\mathcal{A}(G)$ can be efficiently computed from $\mathcal{A}({G'})$. We will call such reductions \emph{Direct-WTERs}. Direct-WTERs are easy to construct and analyze, they do not rely on heavy algorithmic machinery, and they create only one expander instance. Moreover, if a problem admits a Direct-WTER then we can assume that the worst-case graphs to this problem are already expanders. It is therefore futile to invest time trying to solve the problem using the expander decomposition method, since applying expander decompositions on worst-case graphs (which are expanders), may as well simply return the graph itself. Effectively, if the problem admits a Direct-WTER then the answer to Question \ref{q:q2} is negative.

In Direct-WTERs, one should try to optimize the following parameters: 1) the conductance $\phi$ of the output graph $G'$, 2) the running time of transformation, and 3) the \emph{blowup} in the size of $G'$, denoted by $(N,M) = (|V(G')|,|E(G')|)$. We remark that computing $\mathcal{A}(G)$ from $\mathcal{A}(G')$ will always be trivial in our examples (e.g., when $\AG(G)$ is a linear function of $\AG(G')$), so we do not consider it another parameter. The ``gold standard'' for Direct-WTERs is to obtain $\phi = \Omega(1)$, near-linear running time and near-linear blowup (meaning that $N = \tilde{O}(n)$ and $M = \tilde{O}(m)$). If a problem admits such gold standard reduction and it can be solved on $\Omega(1)$-expanders in $O(a(n,m))$ time, for some polynomially bounded function $a(n,m)$ (but not linear), then we can also solve on worst-case graphs in $\tilde{O}(a(n,m))$ time.

However, for problems whose best algorithms run in exponential time, we have to be more stringent in the efficiency requirements: even a blowup of $N = 2n$ weakens the result significantly because it means that to get any major speed-up over $2^{\alpha n}$ time on worst-case graphs (for some $\alpha >0$), one needs to gain a major speed-up over $2^{\alpha n/2}$ in expanders. Namely, gold standard reductions do not necessarily answer Question \ref{q:q1} for exponential-time problems. We can thus speak about a ``platinum standard'', where in addition to the requirements of the gold standard, we also have $N=n + o( n)$. 
If a problem admits such platinum standard reduction and it can be solved in time $2^{\alpha (1-\eps) N } \cdot N^{O(1)}$ on $\Omega(1)$-expanders (for any $\eps > 0$), then we can solve this problem in time $2^{\alpha (1-\eps) (n+o(n)) }\cdot (n+o(n))^{O(1)} = 2^{n (\alpha (1-\eps) + o(1))} \cdot n^{O(1)}$ on worst-case graphs.

Notably, all the Direct-WTERs in our results achieve $\phi = \Omega(1)$ and near-linear running time (but not necessarily near-linear blowup), so we keep the definition clean by not considering them as parameters:
\begin{definition}[Direct-WTER]
\label{def:direct}
A Direct-WTER to problem $\AG$ with blowup $(N,M)$ is a randomized, near-linear time transformation that takes an input graph $G$ and outputs an $\Omega(1)$-expander graph $G'$ w.h.p., such that $|V(G')| = N$, $|E(G')| = M$, and $\AG(G)$ can be computed in constant time from $\AG(G')$.
\end{definition}

In Section~\ref{sec:apps} we present our results which include very efficient Direct-WTERs for Triangle Detection, Maximum Cardinality Matching, and Vertex Cover, proving that the answer for each of them to Question~\ref{q:q1} is positive and to Question~\ref{q:q2} is (unfortunately) negative.
For other problems, however, such as $4$-Cycle, the only Direct-WTERs we could come up with have a super-linear blowup. 
In such cases, we ask whether a (not direct) WTER that does exploit expander decompositions can be obtained.


\paragraph{ED-WTERs}
Before continuing with this discussion, we need to state formally what is an expander decomposition:
\begin{theorem}[Theorems 5.1 and 5.3 in~\cite{Kaplan22}]
\label{thm:ed}
Given a graph $G = (V,E)$ of $m$ edges and a parameter $\phi$, there is a randomized algorithm that with high probability finds a partitioning of $V$ into $V_1,V_2,\ldots,V_k$ such that for all $i \in [k]$, the induced subgraph $G[V_i]$ has conductance $\phi(G[V_i]) \geq \phi$, and $\sum_{i=1}^k \delta(V_i) = O( \phi m \log^2 m)$. The running time of the algorithm is $O(m \log^7 m + \frac{m \log^4 m }{ \phi})$.
\end{theorem}
\noindent
Now, the way that an expander decomposition based reduction works can be defined (loosely) in the following way:
\begin{definition}[ED-WTER]
\label{def:ed-wter}
Problem $\AG$ admits an \emph{expander decomposition-WTER} with conductance $\phi(n)$ and running time $t(n,m)$ (abbreviated as ED-WTER), if there exists an algorithm with oracle access to $\AG(A)$ that runs in time $t(n,m)$, that:
\begin{enumerate}
\item Applies the algorithm from Theorem \ref{thm:ed} to decompose the graph into $\phi(n)$-expanders $G[V_1],\ldots,G[V_k]$.
\item Computes a list of $\Omega(\phi(n))$-expanders $G_1,G_2,\ldots,G_{\ell}$ for some $\ell$. Typically, $\ell=k$ and each $G_i$ is obtained from $G[V_i]$ by a simple procedure (possibly $G_i = G[V_i]$). 
\item Makes oracle calls to $\AG$ on $G_1,\ldots,G_\ell$, and receives the solutions $\AG(G_1),\ldots,\AG(G_\ell)$.
\item Does some computation on the decomposed graph, and together with $\AG(G_1),\ldots,\AG(G_\ell)$, it computes the solution $\AG(G)$.
\end{enumerate}
\end{definition}
ED-WTERs are not captured by Definition \ref{def:direct} mainly because they do not produce only one expander instance $G'$, but many: $G_1,G_2,\ldots,G_\ell$. In fact, many fine-grained reductions work like this. In Section \ref{sec:definition}, we provide a formal definition of fine-grained WTERs that captures both definitions, but for the sake of clarity we will use only Definitions \ref{def:direct} and \ref{def:ed-wter} throughout the paper.

There are three disadvantages to ED-WTERs compared to Direct-WTERs. 
First, they tend to be more complicated and use heavier machinery.
Second, unlike Direct-WTERs, they do not give a negative answer to Question~\ref{q:q2}; if someone discovers a breakthrough algorithm on expanders then (without a Direct-WTER) the way to solve the worst-case problem is to use expander decompositions.
Third, and perhaps most importantly, ED-WTERs cannot produce \emph{true} expanders in the sense that $\phi=\Omega(1)$ but are only limited to proving hardness for graphs with $\phi=(\log{n})^{-O(1)}$.
This is because of the log factors in the expander decompositions (that provably cannot be avoided~\cite{SW19,alev2017graph})
which means that given a graph with expansion $\phi=1/\log{n}$ we cannot expect a decomposition algorithm to decompose it further into true expanders with $\phi=\Omega(1)$.

From the viewpoint of Question~\ref{q:q1}, a WTER that achieves $\phi=(\log{n})^{-O(1)}$ (or even $\phi=\Omega(n^{-\eps})$ for $\eps \to 0$, as we do in this work) is partially satisfying.
On the one hand, it does not show that the problem remains hard on true expanders; it is conceivable that a problem can be solved in $O(2^{1/\phi}\cdot n)$ time which would be linear when $\phi=\Omega(1)$ but inefficient on the class of graphs produced by an ED-WTER.
On the other hand, worst-case graphs have $\phi=O(1/n)$ and graphs with larger $\phi$ are already expander-like. 
Indeed, essentially all breakthroughs using the expander-decompositions method were obtained by exploiting the structures of ``expanders'' with $\phi=\Omega(n^{-\eps})$.

Our results are presented in Section~\ref{sec:apps}, and they include WTERs of both types. Let us conclude this section with two remarks.

\begin{itemize}

    \item Our framework is general with respect to the definition of an expander (e.g. vertex vs. edge expansion, with or without demands, etc.), and it could even be applied more broadly to any graph with strong structural properties (e.g. a Szemer\'edi \emph{regular} graph).
    The specific reductions, however, are sensitive to the exact definition and require modifications to satisfy each definition. 

    \item The reductions we design in this work are \emph{randomized}. This does not affect the message from our results because the problems we consider appear to be just as hard for randomized algorithms as well. 
\end{itemize}

\subsection{Applications}
\label{sec:apps}

\begin{table}[h!]
\begin{small}
\centering
\resizebox{\columnwidth}{!}{
\begin{tabular}{|c|c|c|c|c|}
\hline 
Problem & Best Upper Bound & $\phi$-Expander Lower Bound & $\phi$  \\ \hline 
Triangle Detection (dense graphs) & $n^{\omega}$ & $n^{\omega}$ & $\Omega(1)$ \\ \hline
Triangle Detection (sparse graphs) & $m^{\frac{2\omega}{\omega+1}}$ & $m^{\frac{2\omega}{\omega+1}}$ & $\Omega(1)$ \\ \hline
$k$-Clique & $n^{\omega k/3}$ & $n^{\omega k/3}$ & $\Omega(1)$   \\ \hline 
$4$-Cycle ($m=\Theta(n^{1.5})$) & $n^{2}$ & $n^{4/3}$ & $\Omega(1)$   \\ \hline 
$4$-Cycle ($m=\Theta(n^{1.5})$) & $n^{2}$ & $n^2$ & $\Omega(n^{-\eps})$   \\ \hline 
Subgraph Isomorphism (without pendant vertices)& $n^{f(H)} $ & $n^{f(H)/2}$ & $\Omega(1)$   \\ \hline 
Subgraph Isomorphism (without pendant vertices)& $m^{g(H)}$ & $m^{g(H)}$ & $\Omega(1)$   \\ \hline 
Maximum Cardinality Matching (dense graphs)& $n^{\omega}$ & $n^{\omega}$ & $\Omega(1)$   \\ \hline 
Maximum Cardinality Matching (sparse graphs)& $m \cdot \sqrt{n}$ & $m \cdot \sqrt{n}$ & $\Omega(1)$   \\ \hline 
Vertex Cover (bounded degree) & $1.2125^n \cdot n^{O(1)}$ & $1.2125^{n}/n^{O(1)}$ & $\Omega(1)$    \\ \hline 
Vertex Cover (bounded degree, parameterized) & $1.2738^k \cdot n^{O(1)}$ & $1.2738^{k}/n^{O(1)}$ & $\Omega(1)$    \\ \hline 
Minimum Dominating Set & $1.4969^n \cdot n^{O(1)}$ & $1.4969^{n/(1+\eps)} \cdot n^{O(1)}$ & $\Omega(1)$   \\ \hline 
\end{tabular}
}
\caption{In this table we present our results with respect to the current best upper bounds for the problems. The third column contains the hardness results on $\phi$-expander. In particular, any major improvement to the running times in the third column on $\phi$-expanders will result in a major improvement to the running times of the first column on worst-case graphs. For clarity, we hide any poly-logarithmic factors. }
\label{table:results}
\end{small}
\end{table}

Let us now present the collections of WTERs designed in this paper. 


\paragraph{Triangle Detection} The longstanding upper bound for detecting a triangle in an $n$-node, $m$-edge graph is $O(\min \{n^{\omega},m^{\frac{2 \omega}{\omega+1}}\}) $ \cite{AYZ97} where $\omega<2.38$ is the fast matrix multiplication exponent \cite{AlmanW20}.
Even if $\omega=2$, the upper bound would only be $m^{4/3}$ on sparse graphs. There are different hardness assumptions of different strengths, stating that Triangle cannot be solved much faster, e.g. in $m^{1+o(1)}$ or $m^{4/3-\eps}$ time (see \cite{DBLP:conf/focs/AbboudW14}). The next theorem shows that any polynomial speed-up in the expander-case will break these assumptions in worst-case graphs.
\begin{theorem}[Direct-WTER for Triangle Detection]
Triangle Detection admits a Direct-WTER with blowup $(N,M) = (O(n),\tilde{O}(m))$.
\end{theorem}
In fact, our result for Triangle Detection is merely a special case of a result for the more general problem of detecting $k$-clique in a graph, known simply as the $k$-Clique problem. The best known upper bound for $k$-Clique is $O(n^{\omega k/3})$ when $k$ is divisible by $3$. Otherwise the upper-bound is slightly larger. One of the strongest conjectures in fine-grained complexity (meaning least likely to be true) is that we cannot do better (see \cite{ABW18}). For clarity, let us assume that $k\geq3$ is constant. Then we can state our result in the following way: 
\begin{theorem}[Direct-WTER for $k$-Clique]
\label{thm:k-clique}
For all constants $k \geq 3$, there is a Direct-WTER with blowup $(N,M) = (O(n),\tilde{O}(m))$.
\end{theorem}
For completeness, we remark that if $k = \omega(1)$, then the blowup would be $(N,M) = (O(nk),\tilde{O}(mk^2))$ and the conductance $\phi = \Omega(1/k^2)$. 

Another problem that is closely related to Triangle Detection is $4$-Cycle, for which achieving tight bounds has been notoriously difficult. The longstanding upper bound is $O(\min\{n^2,m^{4/3}\})$ \cite{AYZ97}. Back in 1999, it had been conjectured that a subquadratic algorithm is not possible \cite{YZ97}, but only very recently a super-linear lower bound of $\Omega(m^{1.11})$ (assuming the hardness of Triangle Detection) was obtained \cite{ABKZ22}. 
We remark that the hard cases of $4$-Cycle are in the regime of $m = O(n^{1.5})$, as a simple counting argument shows that any graph with at least $n^{1.5}$ edges contains a $4$-Cycle. 

For $4$-Cycle we design both a Direct-WTER and also an ED-WTER. Nonetheless, both WTERs are relevant because each has a different drawback.
\begin{theorem}[Direct-WTER for $4$-Cycle]
\label{thm:4-cycle}
There is a Direct-WTER to $4$-Cycle with blowup $(N,M) = (\tilde{O}(m),\tilde{O}(m))$.
\end{theorem}
The drawback of the Direct-WTER is a blowup in the number of vertices ($\tilde{O}(m)$ instead of $O(n)$), significantly weakening the result. In particular, worst-case graphs with $m = \Theta(n^{1.5})$ edges are transformed to expanders with $N = \Omega(n^{1.5})$ vertices. Therefore, in order to gain a speed-up of $\tilde{O}(n^{2-\eps})$ on worst-case graphs, one needs a major speed-up of $\tilde{O}(N^{\frac{4}{3}(1-\eps/2)}) = \tilde{O}(n^{2-\eps})$ on $\Omega(1)$-expanders. In search of WTERs that avoid this major drawback, we manage to obtain the next ED-WTER: 

\begin{theorem}[ED-WTER for $4$-Cycle]
\label{thm:ed-wter}
For any $\eps >0$, there is a ED-WTER for $4$-Cycle that runs in time $\tilde{O}(n^{\eps} m + n^{2-\eps/2})$, with conductance $\phi = n^{-\eps}$.
\end{theorem}
This WTER suffers from the disadvantages of ED-WTERs discussed in Section~\ref{sec:contributions}. However, it provides much better hardness results, in the sense that any $\tilde{O}(N^{2-\delta})$ speed-up on $\tilde{\Omega}(n^{-\eps})$-expanders will imply a $\tilde{O}(n^{2-\delta})$ speed-up on worst-case graphs. An interesting open question to ask now is whether $4$-Cycle does become easy on $\Omega(1)$-expanders while it remains hard on $\Omega(n^{-\eps})$-expanders. 

\noindent
We now conclude this line of problems with a general result for the \emph{Subgraph Isomorphism} problem.
\paragraph{Subgraph Isomorphism}
The foregoing problems were special cases of the Subgraph Isomorphism problem that asks to detect a $k$-node pattern graph $H$ inside a host graph $G$ (not necessarily as an induced subgraph). The complexity of solving Subgraph Isomorphism can vary wildly depending on $H$. For example, while the best algorithms for $k$-Clique run in time $O(n^{\omega k/3})$, there are much more efficient algorithms for problems such as $k$-Cycle and $k$-Path, that run in $O(n^{\omega})$~\cite{alon1995color}. Thus, for every pattern $H$ we define a constant $f(H)$ (respectively, $g(H)$) that is the minimum number such that Subgraph Isomorphism (with respect to $H$) can be solved in ${O}(n^{f(H) + o(1)})$ time (respectively, ${O}(m^{g(H) + o(1)})$ time in sparse graphs). We will focus on a variant of Subgraph Isomorphism in which $H$ does not contain pendant vertices (vertices whose degree is one). We note that this variant is still general enough to capture the hard and interesting cases such as $k$-Clique and $k$-Cycle. Our result for this problem is stated by the next theorem:
\begin{theorem}[Direct-WTER for Subgraph Isomorphism without pendant vertices]
\label{thm:subgraph-isomorphism}
There is a Direct-WTER for Subgraph Isomorphism for patterns that do not contain pendant vertices. Assuming that $k$ is constant, the blowup is $(N,M) = (\tilde{O}(m),\tilde{O}(m))$.
\end{theorem}
We should note here that the Direct-WTER for $4$-Cycle is, in fact, simply an application of the Direct-WTER for Subgraph Isomorphism. For completeness, we remark that for non-constant $k$ the blowup is $(N,M) = (\tilde{O}(mk),\tilde{O}(mk))$ and the conductance is $\Omega(1/k)$. As we've noted in the case of $4$-Cycle, there is a significant drawback that comes from the blowup to the number of vertices, and that weakens the result significantly.

\paragraph{Maximum Cardinality Matching}
Next, we turn our attention to the fundamental problem that asks to find a set $M \subseteq E(G)$ of maximum cardinality, such that no two edges in $M$ intersect. 
In a recent breakthrough, an almost-linear time algorithm for Max-flow~\cite{MFlinear22} resulted in a $\tilde{O}(m)$ time algorithm for Maximum Matching in bipartite graphs.
However, the best-known upper bounds for general graphs are still far from linear. For sparse graphs, there is a longstanding $O(m \sqrt{n})$ upper bound by Micali and Vazirani~\cite{MV80}. 
For dense graphs, we can do slightly better using an $\tilde{O}(n^{\omega})$ algorithm by Mucha and Sankowski~\cite{MS04}. 
The Micali and Vazirani algorithm is part of a long list of algorithms that work by finding \emph{augmenting paths} that improve a non-maximum matching iteratively.
A well-known fact says that a matching is non-maximum if and only if there exists an augmenting path in the graph. 

An interesting line of work focused on the case where $G$ is a random graph drawn from the distribution $G(n,p)$. 
Motwani~\cite{motwani1994average} showed that if $p$ is at least $\ln n/ (n-1)$, then with high probability, every non-maximum matching admits a short augmenting path of length $O( \log n )$, 
and therefore the Micali and Vazirani algorithm terminates in $\tilde{O}(m)$ time. Using similar arguments, 
Bast et al.~\cite{bast2006matching} provided a simpler proof that the above holds for any $p \geq 33/n$.
At the heart of their proofs, both papers rely on various properties of random graphs, most notably on vertex-expansion of sufficiently large sets of vertices. 

A question that comes to mind is whether expansion (either vertex-expansion or other forms of expansion) alone is sufficient to make the problem easy, 
or if other properties of random graphs are necessary.  If expansion alone is enough, perhaps we can use the expander decomposition method to improve the state-of-the-art on worst-case graphs. We note that the distinction between vertex-expansion and conductance (which is the notion used in the expander decomposition) does not seem to be very important for their analysis, as they rely on the fact that after expanding two subgraphs of $G$ for $O(\log n )$ steps, the two must meet. This holds even in large-conductance graphs.

Alas, the next Direct-WTER shows that the expander decomposition method is useless against this problem.
\begin{theorem}[Direct-WTER for Maximum Cardinality Matching]
\label{thm:mcm}
There is a Direct-WTER for Maximum Cardinality Matching with blowup $(N,M) = (O(n),\tilde{O}(m))$.
\footnote{To further emphasize that the distinction between vertex-expansion and conductance is not very important, we note that it is possible to show that the graph output by this WTER is also a vertex-expander.}
\end{theorem}

To answer the question about the properties that make random graphs easy for matching algorithms, we look closely at the analysis of Motwani~\cite{motwani1994average} and of Bast et al.~\cite{bast2006matching}. We find that they rely on the fact that with high probability, every pair of large, disjoint sets of vertices are connected by at least one edge.\footnote{See~\cite[Lemma 6.1]{motwani1994average} and~\cite[Lemma 3]{bast2006matching}.} In our reduction, we unavoidably introduce linear-sized sets of vertices that don't have an edge between them. Indeed, the graphs that are output by the reduction may not admit short augmenting paths if the original graph did not admit such. In particular, one cannot use the reduction as a way to speed-up algorithms on worst-case graphs. It is an interesting open question whether such approach, augmenting the graph to be random-like and ``shortcut'' augmenting paths, can be used to improve the running times on worst-case graphs.

We now move on to NP-hard problems, where we are interested in the exact (fine-grained) exponential time complexity.
\paragraph{Minimum Vertex Cover}
In the Minimum Vertex Cover problem, we wish to find a minimum-sized set of vertices such that every edge is incident to this set. This NP-Hard problem has received a lot of attention over the years, resulting in relatively fast algorithms (compared to other NP-hard problems). After a long line of works (e.g. \cite{beigel1999finding,FGK06,razgon2009faster}), the upper bound for Minimum Vertex Cover stands at $1.2125^n\cdot n^{O(1)}$ \cite{BEPR12} and it is unclear if it can be improved further, e.g. to $1.001^n\cdot n^{O(1)}$. The problem has also received a lot of attention from the \emph{parameterized complexity} community, and the current best bound for deciding if there is a vertex cover of size $\leq k$ is $1.2738^k\cdot n^{O(1)}$ \cite{chen2006improved}. In the next theorem we present WTERs that address both the exact complexity and the parameterized complexity of this problem on expanders.
\begin{theorem}[Direct-WTER for Minimum Vertex Cover]
\label{thm:mvc}
There is a Direct-WTER for Minimum Vertex Cover with blowup $(N,M) = (n +O(\log n + \Delta(G)) ,\tilde{O}(m))$, where $\Delta(G)$ is the maximum degree in $G$. Moreover, the size of minimum vertex cover in the output $G'$ is $K=k+ 5 \max\{ \Delta(G), C \log n\}$ for some fixed constant $C$, where $k$ is the size of the minimum vertex cover in $G$.
\end{theorem}
At first sight, this WTER may appear unattractive because of the $\Delta(G)$ blowup in the number of vertices, which could be as large as $\Omega(n)$. However, by observing that the Vertex Cover problem becomes \emph{easier} (in the sense that the $1.2125^n$ bound can be broken) in graphs with large degrees \cite{BEPR12}, we conclude that the hard cases have $\Delta(G)=O(1)$. Therefore, assuming that $\Delta(G) = O(1)$, the blowup is actually $(N,M)=(n+O(\log n),\tilde{O}(m))$ which achieves the so-called platinum standard. 
Furthermore, since the size of the minimum vertex cover in $G'$ is $K = k+5 \max\{ \Delta(G), C \log n \} = k+ O(\log n)$, 
 then any speed-up of $(1.2738-\epsilon)^K \cdot N^{O(1)}$ on $\Omega(1)$-expanders will result in a speed-up of $(1.2738-\epsilon)^{k + O(\log n)} \cdot n^{O(1)} = (1.2738-\epsilon)^{k} \cdot n^{O(1)}$ on worst-case graphs. We see this as a proof of concept that our framework can extend even beyond the traditional notions of complexity and have implications also in parameterized complexity.

Finally, let us present another fundamental NP-Hard problem. %
\paragraph{Minimum Dominating Set}
In the Minimum Dominating Set problem, the goal is to find a set of vertices of minimum size, such that every vertex in the graph is either a neighbor of some vertex in the set or is in the set itself. The Minimum Dominating Set problem can be solved trivially in $2^n \cdot n^{O(1)}$ time. In 2004, a first $c^n \cdot n^{O(1)}$ algorithm with $c<2$ was discovered by Fomin~\cite{fomin2004exact}, and since then, the complexity was further improved to $1.4969^n \cdot n^{O(1)}$~\cite{van2011exact}. 
Unlike Minimum Vertex Cover, the Minimum Dominating Set problem is not fixed-parameter tractable, meaning that it is unlikely to be solved in $f(k) \cdot n^{O(1)}$ time, and in that sense it is considered harder. 

For Minimum Dominating Set we show a Direct-WTER with blowup $(N,M) = ((1+\eps)n,\tilde{O}(m))$ for any constant $\eps >0$, and constant conductance that depends on $\eps$.
Observe that for any significant improvement to the exponent in the running time on $o(1)$-expanders, we can choose $\eps$ to be small enough to get an improvement on worst-case graphs. 
Interestingly, this WTER requires a slightly different approach that is a combination of the techniques we used for $k$-Clique and Maximum Cardinality Matching, together with a subsampling trick. We find it interesting and see it as an indicator that our techniques can be applied to more open questions.
\begin{theorem}[Direct-WTER for Minimum Dominating Set]
\label{thm:mds}
There is a Direct-WTER for Minimum Dominating Set with blowup $(N,M) = (n + \eps n, \tilde{O}(m))$ and conductance $\phi(G') \geq \Omega(1)$, for any constant $\eps >0$.
\end{theorem}

\medskip
Our results only reveal the tip of the iceberg. There is a plethora of problems for which Questions \ref{q:q1} and \ref{q:q2} may be relevant, from open questions in fine-grained complexity to NP-hard problems. We are hopeful that our framework will help in understanding the role that expanders have in algorithms and beyond.

\subsection{Related Work}

An interesting line of research in fine-grained complexity it to show conditional lower bounds for ``realistic'' input classes such as planar graphs, bounded treewidth graphs, and so on. There has been some progress on designing reductions that apply to such restricted families as well (see e.g. \cite{AVW16,AD16,EvaldD16,ACK20}). Perhaps the next extension to this line of work should be random-like graphs, e.g. expanders? Indeed, a recent independent paper is concerned with the complexity of various problems on families of dynamic graphs, that also include expanders~\cite{henzinger2022fine}. Our work is not motivated by these kinds of questions, but our techniques may have implications for that line of work as well.

Another area of research, that somewhat resembles ours, is the research on worst-case to average-case reductions. These reductions show that a problem is as easy on worst-case inputs as it is on \emph{random} instances (e.g., taken from the uniform distribution). In recent years, much effort has gone into proving worst-case to average-case reductions for problems in P (see \cite{BRSV17, DLW20, HLS22, asadi2022worst}). However, it has only been accomplished for counting or algebraic problems but not decision problems, while we are interested in graph decision problems on pseudo-random instances. In fact, trying to show any worst-case to average-case reductions on these problems might as well fail, since they (e.g., Triangle Detection) typically do become easier on random graphs.

\section{Preliminaries}
Let $G = (V,E)$. For every $v \in V(G)$ define the neighborhood of $v$ to be $N(v) := \{ u \in V \mid \{u,v\} \in E(G) \}$.  
The degree of a vertex $v \in V(G)$ is denoted by $\deg_G(v) := |N(v)|$, and for $u \notin V(G)$ define $\deg_G(u) := 0$. We denote the maximum degree among the vertices of $G$ by $\Delta(G)$. 
For every $S \subseteq V$, denote the \emph{volume} of $S$ by $vol_G(S) = \sum\limits_{v \in S} \deg_G(v)$. We say that an edge $\{u,v\}$ is an \emph{internal} edge is both $u$ and $v$ are in $S$. Denote by $E_{G}(S,S')$ the set of edges between disjoint sets of vertices $S,S' \subseteq V(G)$ and by $e_G(S,S') := |E_G(S,S')|$. Denote by $\delta_G(S) := e_G(S,\bar{S})$ and note that $\delta(S) = \delta(\bar{S})$. The \emph{conductance} of $S$ is defined as $\phi_G(S):=\frac{\delta_G(S)}{\min\{vol_G(S),vol_G(\bar{S})\}}$ and for $S = \emptyset$ we define $\phi_G(S) := 1$. The conductance of $G$ is defined as $\phi(G) = \min_{S \subseteq V} \phi_G(S)$. The \emph{expansion} of a cut or a graph is simply its conductance.  We omit the subscripts when they are clear from the context. We denote by $G[S]$ the induced subgraph of $G$ on set of vertices $S$. Throughout that paper, we use the Chernoff bound, stated as:
\begin{fact}[Chernoff Bound]
Let $X$ be the sum of $n$ independent indicators, each taking the value $1$ with probability $0<p<1$, and $\mu$ be the expectation of $X$. Then for any $0\leq \delta \leq 1$, 
$
\p[|X - \mu| \geq \delta \mu] \leq \exp({-{\delta^2 \mu}/{3}}) .
$
\end{fact}

\section{Technical Overview}

\label{sec:overview}

Recall that our goal when designing a Direct-WTER is to take a worst-case graph $G$ and turn it into an expander $G'$ without changing the answer to the problem uncontrollably. 
Consider the Triangle Detection problem as an example for now. The goal is to add edges to a graph in order to make it an expander while not introducing a triangle to the graph if the original graph was triangle-free.

Let us start by recalling the intuition that an expander graph should be \emph{mixing} \cite{oded_expander} in the sense that a short random walk from any vertex is equally likely to reach any other vertex. Keeping this in mind will help the reader follow the informal arguments in this section, but it is not necessary for understanding our actual proofs that only work with the definition of expander in Definition~\ref{def:expander}.

The first idea that comes to mind is to take an expander graph $X$ on $n:=|V(G)|$ nodes (e.g. a random sparse graph) and \emph{add it} to $G$, so that $G'=G \cup X$.
There are several issues with this idea. 
First, the resulting graph will not be an expander if $X$ is just an arbitrary expander.
For example, if $G$ is dense and $X$ is sparse, then the new edges of $X$ will not have much effect on random walks in the combined graph $G'$.
This issue can be resolved by using an expander $X$ in which each vertex $v$ has the same degree that it has in $G$; a random graph with specified degrees does the job.
The second issue is that the answer to the problem can change arbitrarily when taking the union with another graph $X$; e.g. an edge $\{u,v\} \in E(G)$ may become a triangle because $u,v$ have a common neighbor in the newly introduced $X$.
For some problems, e.g., the Subgraph Isomorphism problem for pattern $H$ without degree-$1$ vertices, this issue can be overcome by \emph{subdividing} the edges of $X$ into paths of length $|H|+1$ before adding them to $G$. 
This makes the new edges useless for forming any copy of $H$ (or a triangle in particular) because any subgraph of size $|H|$ (or $3$) that uses a new edge must have a leaf node of degree $1$. 
This does result in a Direct-WTER, but it is not as fine-grained as we would like: the number of nodes in $G'$ becomes $O(m)$ where $m:=|E(G)|$ because of the subdivided edges.
Our actual reductions use additional ideas and are actually quite different, but the intuition behind them is similar. 

Instead, the basic idea in our reductions is to add a small \emph{expansion layer} $U$ with $O(n)$ nodes and connect each node in $V(G)$ with random edges to $U$. That is, we add a random bipartite graph whose one side is $V(G)$, and the other side is $U$.
If a node $v \in V(G)$ has degree $d$ in $G$ then we add $d$ random edges to $U$.
Then, a random walk starting at $v$ will go with probability $1/2$ to the expansion layer $U$, and the next step would send it back to an essentially random vertex $w \in V(G)$; thus simulating a ``random edge'' from $v$ to $w$.
We prove that adding such an expansion layer to any graph that has a sufficiently large (logarithmic) minimum degree makes it an expander. We introduce some gadgets to artificially increase the degrees in our problems.

The advantage of adding an expansion layer over simply taking the union with a random graph is that it makes it easier to control the change in the solution of the problem.
This is accomplished with different gadgetry for each problem and with a different argument for why the gadgets do not harm the conductance of the graph by too much.
For example, in the WTER for maximum matching, we add a unique ``twin'' vertex to each node in the expansion layer $U$ and connect them by an edge. The maximum matching is forced to match every node in $U$ to its twin, thus making the newly added edges between $V(G)$ and $U$ unusable in a maximum matching. The WTER for Vertex Cover uses a similar idea. 
A different idea is used for Triangle (and $k$-Clique), where we connect the expansion layer $U$ only to a (sufficiently large) independent set $V_1$ in $G$. Consequently, the newly added edges cannot form any triangle. The WTER for Minimum Dominating Set combines ideas from both WTERs for Triangle Detection and Vertex Cover, together with a sampling trick. Notably, for the $4$-Cycle problem, we could not avoid the subdivision trick so we complement our result with an ED-WTER that avoids the blowup.

Designing ED-WTERs is fundamentally different. To reduce $4$-Cycle into expanders, we use the expander decomposition to partition $G$ into induced subgraphs with conductance $\phi = n^{-\eps}$ and a small set of outer edges. Then, assuming that we can check (in subquadratic time) whether there is a $4$-cycle inside each subgraph, the problem boils down to finding a $4$-cycle that uses one of the outer edges. Hence, we provide an algorithm that finds a $4$-cycle that uses one of the outer edges or outputs that there is no such cycle. Using the fact that the number of outer edges is small (namely, $O(n^{1.5-\eps})$), this algorithm runs in subquadratic time.

\section{The Reductions}
\label{sec:reduction}
Our main building block is a randomized algorithm that takes as an input an $n$-vertex graph $G$ where for every $v \in V(G)$, $\deg_G(v) \geq C \log n$ for some large enough constant $C$, and outputs a graph $G'$ such that: (1) $|E(G')| = O( |E(G)| )$ and $|V(G')|  =O( n )$, and (2) $\phi(G') = \Omega(1) $. We note that this algorithm applies even to graphs with self-loops. We later modify and use this algorithm in different ways to show all the above results.
\subsection{The Expansion Layer Construction}
\label{sub:basic}
Construct an expander graph $G'$ from $G$ as follows. 
\begin{enumerate}
    \item Add to $G$ an \emph{expansion layer} $U$ that consists of $5n$ vertices.
    \item For every vertex $v \in V(G)$, sample without replacement $\deg_{G}(v)$ vertices from the expansion layer $U$ and add all edges between them and $v$. Thus, the total number of vertices and edges added in this step is $O(n)$ and $O(|E(G)|)$, respectively.
    \item  Increase the degrees of the vertices of $U$ deterministically in the following way. Partition each of $V(G)$ and $U$ into $\frac{10n}{C \log n}$ parts of equal size (i.e, partition $V(G)$ into parts of size $\frac{C}{10} \log n$ and partition $U$ into parts of size $\frac{C}{2} \log n$). Then, take an arbitrary perfect matching between the parts in the partitions of $V(G)$ and $U$, and add a bi-clique between every matched pair of parts (i.e. add all possible edges between every matched pair). Thus, the total number of edges added in this step is $O( \frac{n}{ \log n} \cdot \log^2 n ) = {O}(n \log n)$. Note that after this step, the degrees of all vertices in $U$ are at least $\frac{C}{10} \log n$ and the degrees of all vertices in $V(G)$ are at least $\frac{3}{2} \cdot C \log n $ (as we assume that $\deg_G(v) \geq C \log n$ and then add an additional $\frac{C}{2} \log n$ edges).
\end{enumerate}
See Figure \ref{fig:alg}. Observe that $|E(G')| = O( |E(G)| )$, $|V(G')| = O(n)$ and that the reduction is randomized and runs in linear time. 

\begin{figure}[h]
\centering
\includegraphics[scale=0.8]{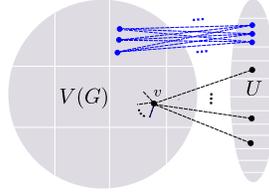}
\caption{The basic construction of the algorithm. The blue edges represent the deterministically added bi-cliques. The black edges represent the edges of $G$ and the new random edges. For each vertex $v$ we sample $\deg_G(v)$ neighbors in $U$ and add edges between them.}
\label{fig:alg}
\end{figure}

\begin{theorem}
\label{thm:basic}
The conductance of $G'$ is $\phi(G') \geq 0.01$ with probability at least $1-O(\frac{1}{n})$. 
\end{theorem}

Before proving this result, let us introduce some notation and give an overview. Throughout the section, we omit the subscripts from $\deg_{G'}(\cdot),vol_{G'}(\cdot),E_{G'}(\cdot,\cdot),e_{G'}(\cdot,\cdot),\delta_{G'}(\cdot)$ and $\phi_{G'}(\cdot)$. For any cut $S \subseteq V(G')$,
we denote by $S_V$ and $S_U$ the parts $S \cap V(G)$ and $S \cap U$, respectively.

To prove the theorem, we will focus on the complement probability. That is, we show that $\p[ \phi(G') < 0.01] = O(\frac{1}{n})$. 
We will apply the union bound over all possible cuts. More precisely, we show that for any cut $S \subseteq V(G')$, the \emph{failure probability} $\p[\phi(S) < 0.01]$ is sufficiently small so that by summing the failure probabilities of all possible cuts we still get a value much smaller than $1$.

For any cut $(S,\bar{S})$, assume w.l.o.g. that $S$ is the part for which the inequality $|S_U| \leq |U|/2$ holds (as it must hold either for $S$ or $\bar{S}$). Observe that in order to show that $\p[\phi(S) < 0.01]$ is small it suffices to show that $\p[\frac{\delta(S)}{vol(S)}<0.01]$ is small, since:
\[
\phi(S) = \frac{\delta(S)}{\min\{vol(S),vol(\bar{S})\}} = \max\left\{\frac{\delta(S)}{vol(S)},\frac{\delta(S)}{vol(\bar{S})}\right\}
\]
and therefore $\p[\phi(S) <  0.01] \leq \p[\frac{\delta(S)}{vol(S)}<0.01]$. We thus prove the next lemma:

\begin{lemma}
\label{lem:basic}
For any cut $S \subseteq V(G')$ such that $|S_U| \leq |U|/2$, we have:
\[
\p\left[ \frac{\delta(S)}{vol(S)} < 0.01 \right] \leq n^{- \frac{C}{800} \cdot |S| }.
\]
\end{lemma}

\begin{proof}
\begin{figure}[b]
    \centering
    \includegraphics[scale=0.8]{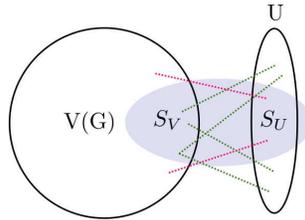}
    \caption{A depiction of a graph $G$ with expansion layer $U$. The gray area denotes a set $S$. In the event $A$ in our proof, we lower bound the number of pink edges, and in the complement event $\bar{A}$, we lower bound the number of green edges.}
    \label{fig:my_label}
\end{figure}

We split into cases by using the law of total probability. Denote by $A$ the event that ${ vol(S_V) < \frac{vol(S_U)}{2}}$. Then:
\begin{equation*}
\p\left[ \frac{\delta(S)}{vol(S)} <0.01 \right] = \p\left[ \frac{\delta(S)}{vol(S)} < 0.01 \middle| A \right] \cdot \p[A] + \p\left[\frac{\delta(S)}{vol(S)}  < 0.01  \middle|\bar{A} \right]\cdot \p[\bar{A}] .
\end{equation*}
We claim that in the event of $A$, the value of $\frac{\delta(S)}{vol(S)}$ must be at least $1/3$, and therefore the first summand equals $0$.
To see why, first observe that:
\begin{equation}
\label{eq:su}
vol(S_U) = e(S_U,S_V) + e(S_U,V(G) \setminus S_V) \leq vol(S_V) + \delta(S) .
\end{equation}
Therefore, 
$
vol(S) = vol(S_V) + vol(S_U) \leq 2 vol(S_V) + \delta(S), 
$
and hence 
\[
\frac{\delta(S)}{vol(S)} \geq \frac{\delta(S)}{2 vol(S_V) + \delta(S)} .
\]

Now, since we are assuming the event $A$ we also have $2 vol(S_V) < vol(S_U)$. Therefore, together with Equation \ref{eq:su} we get $vol(S_V) \leq \delta(S)$. Hence:
\[
\frac{\delta(S)}{2 vol(S_V) + \delta(S)} \geq \frac{\delta(S)}{3 \delta(S)} = \frac{1}{3} > 0.01 .
\]

We thus conclude:
\begin{equation*}
\begin{split}
&\p\left[ \frac{\delta(S)}{vol(S)}  < 0.01 \right] = 0 \cdot \p[A] + \p\left[\frac{\delta(S)}{vol(S)}  < 0.01  \middle| \bar{A} \right]\cdot \p[\bar{A}] \leq 
\p\left[\frac{\delta(S)}{vol(S)} < 0.01  \middle| \bar{A} \right] .
\end{split}
\end{equation*}
For the rest of the proof, we deal with bounding the right term; i.e. the failure probability in the event $\bar{A}$ where $vol(S_V) \geq \frac{vol(S_U)}{2}$. First, let us refer to $S_V$-to-$(U \setminus S_U)$ edges as \emph{nice} edges. Note that nice edges belong to $\delta(S)$.

We show that with high probability, in the event of $\bar{A}$, the number of nice edges is large compared to $vol(S)$. Consider the $\deg_G(v)$ edges sampled by some $v \in S_V$ during the construction of $G'$. Denote by $X_v$ the number of edges that got sampled by $v$ and are nice, and let $X = \sum_{v \in S_V} X_v$ be the total number of nice edges that got sampled by the vertices of $S_V$.
 We note that the total number of edges sampled by the vertices of $S_V$ (either nice or not) is $\sum_{v \in S_V} \deg_G(v) = vol_G(S_V)$. It will be beneficial for us to view the sampling of edges as a sequential process. That is, every $v \in V(G)$ samples $\deg_G(v)$ vertices from $U$, one by one and without replacement. Since the sampling is done without replacement, then after each step, the number of vertices in $U$ that hadn't been sampled yet is reduced by one. We call those vertices \emph{valid} vertices. 
Note that an edge that got sampled by $v$ is nice if the corresponding vertex that got sampled by $v$ is in $U \setminus S_U$. We claim that the probability that a sampled edge is nice is at least $0.1$, even when conditioned on $\bar{A}$. This is formalized by the next claim.
\begin{claim}
Let by $u_i \in U$ be the $i^{th}$ vertex that got sampled by $v$. Then, 
$
\p[ u_i \notin S_U \mid \bar{A} ] \geq 0.3 .
$
\end{claim}
\begin{proof}
First consider the above probability but without the condition of $\bar{A}$ (i.e, the probability $\p[u_i \notin S_U]$). The number of valid vertices in $U \setminus S_U$ when we sample $u_i$ is at least $|U| - |S_U| - (i-1)$. 
Recall that we assume that $|S_U| \leq 2.5n$. Since $i \leq \deg_G(v) \leq n$ and $|U| = 5n$ then this value is at least $5n - 2.5n - n = 1.5n$. Hence, 
$\p[ u_i \notin S_U] \geq \frac{1.5 n}{5n} = 0.3 $.

To prove the claim under the condition of $\bar{A}$, we use Bayes formula:
\begin{equation*}
\begin{split}
\p[ u_i\in U \setminus S_U \mid \bar{A} ] = \frac{\p[\bar{A} \mid u_i \in U \setminus S_U ] \cdot \p[ u_i \in U \setminus S_U]}{\p[\bar{A}]} \geq 0.3 \cdot \frac{\p[\bar{A} \mid u_i \in U \setminus S_U ]}{ \p[ \bar{A}] } .
\end{split}
\end{equation*}
Recall that $\bar{A}$ is the event $vol(S_U) \leq 2 vol(S_V)$, therefore $\bar{A}$ is more likely when conditioned on $u_i \in U \setminus S_U$ because it means that one sampled edge does not contribute to $vol(S_U)$.
Therefore, 
$
{\p[\bar{A} \mid u_i \in U \setminus S_U ]}\geq { \p[ \bar{A}] } $.

\end{proof}

Consider now a (separate) random process that consists of $vol_G(S_V)$ independent Bernoulli trials, each of which has a success probability $0.3$. The distribution of the total number of successes $Y$ is the \emph{binomial distribution} with parameters $vol_G(S_V)$ and $0.3$. We want to relate $Y$ to $X$ so that we can apply the strong concentration bounds associated with the binomial distribution on $X$. The problem is that $X$ is not distributed binomially because there are dependencies that arise from the fact that every vertex $v \in S_V$ samples vertices from $U$ \emph{without replacement}. As noted previously, even though samples made by the same vertex are dependent, we are still guaranteed to have a success probability of $0.3$ per sample, even when conditioned on any previous event. Thus, it follows from a coupling argument that $X$ \emph{stochastically dominates} $Y$. That is, for any $k$ we have:  
$
\p[X \geq k | \bar{A}] \geq \p[ Y \geq k ], 
$
and equivalently, we also have 
$\p[X \leq k | \bar{A}] \leq \p[ Y \leq k ]$. 
Therefore, since $\e[Y] = 0.3 \cdot vol_G(S_V)$ then by the Chernoff bound we have for any $0<\epsilon <1$:
\begin{equation*}
\begin{split}
&\p\left[ X \leq (1-\epsilon) \cdot 0.3 \cdot vol_G(S_V) | \bar{A}\right] \leq 
\p[ Y \leq  (1-\epsilon) \cdot 0.3 \cdot vol_G(S_V) ] \leq 
\exp\left(\frac{-\epsilon^2 \cdot \e[Y]}{3}\right) .
\end{split}
\end{equation*}
In particular, it implies (by setting $\epsilon = 0.7$) that:
\begin{equation}
\label{eq:chernoff}
\p\left[ X \leq 0.09 \cdot vol_G(S_V) | \bar{A} \right] \leq \exp\left(-\frac{49 vol_G(S_V)}{300}\right) \leq \exp\left(-\frac{ vol_G(S_V)}{10}\right).
\end{equation}
We've shown that $X$ and therefore $\delta(S)$ is likely to be a constant fraction of $vol_G(S_V)$. It remains to show that $vol_G(S_V)$ is at least some constant fraction of $ |S| \cdot C \log n$, which would imply that the probability in Equation \ref{eq:chernoff} is smaller than $n^{-\Omega(|S|)}$.  Observe that since we assume the event $\bar{A}$ then $vol(S_V) \geq vol(S_U)/2$. Thus, we get 
$
vol(S_V) = \frac{1}{2} \cdot ( vol(S_V) + vol(S_V) ) \geq  
\frac{1}{4} \cdot vol(S_U) + \frac{1}{2} vol(S_V),
$
and by the lower bounds on the degrees, we have:
$
vol(S_V) \geq \frac{C}{40} \cdot |S_U| \cdot \log{n} + \frac{3C}{4} |S_V| \cdot \log{n} .
$

Next, we show that $vol(S_V)$ is upper bounded by some constant times $vol_G(S_V) + |S| \log n$. Observe that:
\begin{equation}
\label{eq:blah}
vol(S_V) \leq 2 vol_G(S_V) + |S_V| \cdot \frac{C}{2} \log n ,
\end{equation}
because every edge that we do not count in $2 vol_G(S_V)$ must come from the deterministically added edges.  Thus,
\begin{equation*}
\frac{C}{40} \cdot |S_U| \cdot \log{n} + \frac{3C}{4} |S_V| \cdot \log{n} \leq vol(S_V) \leq 2 vol_G(S_V) + \frac{C}{2} |S_V| \cdot \log n ,
\end{equation*}
and therefore, we get:
\begin{equation}
\label{eq:volg}
vol_G(S_V) \geq 
(\frac{C}{40} \cdot |S_U| \cdot \log{n} + \frac{C}{4} |S_V| \cdot \log{n})/2
\geq 
\frac{|S| \cdot C \log{n}}{80}.
\end{equation}
Now, from Equations \ref{eq:chernoff} and \ref{eq:volg} we get:
\begin{equation}
\label{eq:chernoff2}
\begin{split}
&\p\left[ X \leq 0.09 \cdot vol_G(S_V) \middle| \bar{A}\right] \leq \exp\left(-\frac{vol_G(S_V)}{10}\right) \leq 
\exp\left(-\frac{\frac{C}{80} \cdot |S| \log n}{10}\right)  = n^{- \frac{C}{800} \cdot |S| }.
\end{split}
\end{equation}
Let us now finish the proof of the lemma. Since we're in the event $\bar{A}$ then
\begin{equation}
\label{eq:blah2}
vol(S) = vol(S_V) + vol(S_U) \leq 3 vol(S_V),
\end{equation}
and since $|S_V| \cdot C \log n \leq vol_G(S_V)$ (by assumption on the degrees) then together with Equation \ref{eq:blah} we have 
$
vol(S_V) \leq 2 vol_G(S_V) + |S_V| \cdot \frac{C}{2} \log n \leq 3 vol_G(S_V) .
$
Therefore, by the fact that $\delta(S) \geq X$ and by Equations \ref{eq:blah2} and \ref{eq:chernoff2}, we get:
\begin{equation*}
\begin{split}
& \p\left[ \frac{\delta(S)}{vol(S)}  < 0.01 \middle| \bar{A} \right] \leq  
 \p\left[ \frac{\delta(S)}{3 vol(S_V)}  < 0.01 \middle| \bar{A} \right] \leq 
 \p\left[ \frac{\delta(S)}{9 vol_G(S_V)}  < 0.01 \middle| \bar{A} \right] \leq \\
 &\p\left[ \frac{X}{9 vol_G(S_V)}  < 0.01 \middle| \bar{A} \right] =
 \p\left[ X  < 0.09 \cdot vol_G(S_V)\middle| \bar{A} \right]  \leq 
n^{- \frac{C}{800} \cdot |S| }
\end{split}
\end{equation*}
\end{proof}

\begin{proof}[Proof of Theorem \ref{thm:basic}]
Fix $C$ to be large enough so that the failure probability in Lemma \ref{lem:basic} is less than $n^{-2|S|}$ ($C = 1600$ will do).
Since the number of different cuts that consist of $k$ vertices is $O(n^k)$, we can apply the union-bound over all cuts such that $|S_U| \leq |U|/2$. In more detail, we sum over $k=1,2, \ldots,3.5 n$ (any cut larger than that must have $|S_U| > |U|/2$) the number of cuts of size $k$ times the failure probability of such cuts. 
Formally, for any non-empty cut $(S,\bar{S})$ with $|S_U| \leq |U|/2$, denote by $B_S$ the event that $\phi(S) < 0.01$. By Lemma \ref{lem:basic} and the union bound, we have:
\begin{equation*}
\begin{split}
&\p[\phi(G) < 0.01] = \p\left[ \bigcup_{\substack{S \subseteq V(G'),\\|S_U| \leq 2.5n}} B_S \right] \leq 
\sum_{\substack{S\subseteq V(G'),\\|S_U| \leq |U|/2}} \p[B_S] \leq 
\sum_{\substack{S\subseteq V(G'),\\|S_U| \leq |U|/2}} n^{-2 |S|} 
\end{split}
\end{equation*}
Note that for any cut $S$ with $|S_U| \leq |U|/2$ we have $|S| \leq 3.5n$. Therefore: 
\begin{equation*}
\begin{split}
&\sum_{\substack{S\subseteq V(G'),\\|S_U| \leq |U|/2}} n^{-2 |S|} =
\sum_{k=1}^{3.5n} \sum_{\substack{|S| = k,\\|S_U| \leq |U|/2}} \cdot n^{-2k} \leq 
\sum_{k=1}^{3.5n} \binom{6n}{k} \cdot n^{-2k} \leq 
\sum_{k=1}^{3.5n} (6n)^k \cdot n^{-2k} = O\left(\frac{1}{n}\right)
\end{split}
\end{equation*}
\end{proof}

Let us now explain why we can assume that the degrees in $G$ are at least $C \log n$. If for some $v \in V(G)$ we have $deg_G(v) < C \log n$, then we can add $C \log n - deg_G(v)$ self-loops to $v$ in a preprocessing step. After constructing the expander $G'$ we can remove all those self-loops from $G'$. Notice that by removing self-loops from an expander we only increase the expansion of the graph.

\subsection{Direct-WTERs for Maximum Cardinality Matching and Minimum Vertex Cover}

We are now ready to show our first WTERs for the problems of computing the Maximum Cardinality Matching (abbreviated as MCM) and Minimum Vertex Cover (abbreviated as MVC) of a graph $G$. We use $MCM(G)$ and $MVC(G)$ to denote their values, respectively.

Nai\"vely, we could try to simply apply Algorithm \ref{sub:basic} on $G$, resulting in an expander graph $G'$. However, this approach will fail 
because there is no obvious way to compute $MCM(G)$ ($MVC(G)$) from $MCM(G')$ ($MVC(G')$). 
Therefore, we employ a simple trick that helps us control the values of the MCM and MVC. The trick is to add \emph{twin vertices} to a selected subset of the vertices of the graph. Formally, we define a twin of a vertex $u$ to be a degree-$1$ vertex $u'$ that is adjacent to $u$ (that is, $u'$ is a pendant vertex). By the next fact, we can add twins to the graph to force a selected subset of edges or vertices to be inside the MCM or MVC, respectively.
\begin{fact}
\label{fact:twins}
If a graph contains a twin vertex $u'$ of $u$, then there exists an MCM that contains the edge $\{u,u'\}$, and there exists an MVC that contains the vertex $u$.
\end{fact}
We note that we need to be careful that by adding twins to an expander graph, we do not reduce the expansion by much. Formally, we prove this in the next lemma:
\begin{lemma}
\label{lem:twins}
Let $H$ be a connected graph and let $A \subseteq V(H)$ be arbitrary. Let $H'$ be a graph constructed from $H$ by adding a twin vertex $u' \in V(H')$ and an edge $\{u,u'\} \in E(H')$ for every $u \in A$. Then, $\phi(H') \geq \phi(H)/2$.
\end{lemma}
\begin{proof}
Let $S \subseteq V(H')$. Observe that $\delta_{H'}(S) \geq \delta_H(S)$. Additionally, since every vertex in $S$ has at most one twin then $\deg_{H'}(v) \leq \deg_{H}(v) + 1$ for every $v \in S$, and therefore $vol_{H'}(S) \leq |S| + vol_H(S)$. Since $H$ is connected, then $|S| \leq vol_H(S)$ and therefore 
$vol_{H'}(S) \leq |S| + vol_G(S) \leq 2 vol_H(S).$ 
Thus, since the above holds for $\bar{S}$ as well, we get:
\[
\phi_{H'}(S) = \frac{\delta_{H'}(S)}{\min\{vol_{H'}(S),vol_{H'}(\bar{S})\}} \geq \frac{\delta_H(S)}{2 \min \{vol_H(S),vol_H(\bar{S}) \} } = \frac{\phi_H(S)}{2}, 
\]
and therefore $\phi(H') \geq \frac{\phi(H)}{2}$.
\end{proof}
\noindent
We are now ready to show a Direct-WTER for MCM.
\begin{proof}[Proof of Theorem \ref{thm:mcm} (Direct-WTER for MCM)]
Apply Algorithm \ref{sub:basic} on $G$ to get an $\Omega(1)$-expander $\hat{G}$. Then add a twin vertex $u'$ for every $u \in U$, where $U$ is the expansion layer in $\hat{G}$. Denote the resulting graph by $G'$. See Figure \ref{fig:mcm}. Note that the number of twins added is $|U| = 5n$ and that the number of added edges is $O(m + n\log n)$. Hence, the blowup is $(N,M) = (O(n),\tilde{O}(m))$. Moreover, by Lemma \ref{lem:twins} we have $\phi(G') = \Omega(1)$.
\begin{figure}[h!]
    \centering
    \includegraphics[scale=0.8]{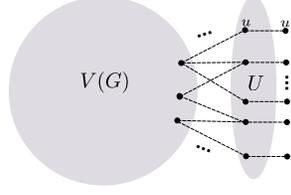}
    \caption{The graph $G'$ obtained from our construction for MCM and MVC.}
    \label{fig:mcm}
\end{figure}

\noindent
Next, we prove that computing $MCM(G)$ from $MCM(G')$ is trivial.
\begin{claim}
The size of the maximum cardinality matching in $G'$ is 
$MCM(G') = MCM(G) + 5n$.
\end{claim}
\begin{proof}
Denote by $X \subseteq E(G')$ the set of all edges incident to twins. By Fact \ref{fact:twins}, there exists in $G'$ a maximum matching $M$ that contains $X$. Notice that $M$ does not contain any $U $-to-$V(G)$ edges and therefore $|M| =  |M \cap E(G)| + |X| = |M \cap E(G)| +5n$. Moreover, $M \cap E(G)$ must be a maximum matching of $G$ because if there was a larger matching we could take it with $X$ to obtain a larger matching for $G'$. Thus, 
$|M| = |M \cap E(G)| + 5n = MCM(G) +  5n$.
\end{proof}
We've shown a transformation from $G$ to a $\Omega(1)$-expander $G'$ with blowup $(N,M) = (O(n),\tilde{O}(m))$, such that $MCM(G)$ can be computed efficiently from $MCM(G')$.
\end{proof}

Next, we show a Direct-WTER for MVC. Since the best upper bounds for MVC are exponential then we aim to keep the blowup in the number of vertices small. 
To this end, we slightly modify Algorithm \ref{sub:basic} to avoid adding an expansion layer of size $5n$.
\paragraph{Modifications to Algorithm \ref{sub:basic}.} Instead of adding an expansion layer of size $5n$, we add an expansion layer of size $5\max\{\Delta(G),C \log n\}$. Notice that $U$ is sufficiently large to allow every vertex in $V(G)$ to sample $\max\{ \deg_{G}(v), C \log n \} $ vertices from $U$ without replacement. 
We claim that Theorem \ref{thm:basic} holds under this modification. To see why, observe that throughout the construction and analysis we only needed $U$ to be large enough so that: (1) every vertex $v \in V(G)$ can sample without replacement $\max\{ \deg(v), C \log n \}$ vertices from $U$, and (2) $|U| - |S_U| - \max\{\deg(v),C\log n\} \geq 1.5 |U|$ which indeed holds as long as $|S_U| \leq 2.5 |U|$. 

\begin{proof}[Proof of Theorem \ref{thm:mvc} (Direct-WTER for MVC)]
Apply the modified Algorithm \ref{sub:basic} on $G$ and add a twin vertex $u'$ for every $u \in U$. Denote the resulting graph by $G'$.
The main difference between the previous construction and this one is that we now have
$
|U| = 5 \max\{ \Delta(G), C \log n \}
$
and therefore, the blowup in $G'$ is $(N,M) = (n + O(\log n + \Delta(G)), \tilde{O}(m))$. Moreover, 
\begin{claim}
The size of the minimum vertex cover in $G'$ is:
\[
MVC(G') = MVC(G) + 5\max\{\Delta(G), C \log n\}
\]
\end{claim}
\begin{proof}
By Fact \ref{fact:twins}, there exists a minimum vertex cover $D$ of $G'$ that contains $U$. Therefore, $D$ is a disjoint union of $U$ and $D_V = D \cap V(G)$. 
Note that $D_V$ must cover every edge of $G$. Moreover, $D_V$ must be a minimum vertex cover in $G$ because any smaller vertex cover $\tilde{D}$ can be added to $U$ to obtain a smaller vertex cover of $G'$. Thus:
\[
MVC(G') = |D_V| + |U | = MVC(G) + 5 \max \{\Delta(G),C \log n\}
\]
\end{proof}
Hence, we've shown a transformation from $G$ to a $\Omega(1)$-expander $G'$ with blowup $(N,M) = (n+O(\log n + \Delta(G)),\tilde{O}(m))$, such that $MVC(G)$ can be computed efficiently from $MVC(G')$.
\end{proof}

\subsection{Direct-WTER for $k$-Clique}
The construction for $k$-Clique can be described roughly in the following way. We apply a simple transformation that makes the graph $k$-partite, and such that every vertex has the same number of neighbors in every part (except the one it belongs to). This property is crucial as it implies that any sparse cut in the graph must contain vertices from all parts. 
 We apply Algorithm \ref{sub:basic} to only one of the parts and obtain a graph $G'$ (we later describe this step in more detail). Since any sparse cut in the $k$-partite graph must contain vertices from the part that has an expansion layer, there are no sparse cuts in $G'$. Specifically, the graph $G'$ has conductance $\Omega(1/k^2)$ and the number of $k$-Cliques in $G'$ is $k!$ times the number of $k$-cliques in $G$. 

\begin{proof}[Proof of Theorem \ref{thm:k-clique} (Direct-WTER for $k$-Clique)]
We would like to assume that $G$ is $k$-partite with parts $V_1,V_2,\ldots,V_k$, such for every vertex $v \in V_i$, $v$ has the same number of neighbors in every part $V_j$ for $j \neq i$. We can take $k$ copies $V_1,\ldots,V_k$ of $V(G)$ such that $V_i = \{v_i \mid v \in V(G)\}$. Now, for every edge $\{u,v\} \in E(G)$ we create $\binom{k}{2}$ copies: $\{ \{u_i,v_j\} \mid i \neq j\}$. Observe that the resulting graph is $k$-partite, the number of vertices in this graph is $kn$, and the number of edges is $k \cdot m^2$. Most importantly, the number of $k$-cliques in this graph is $k!$ times the number of $k$-cliques in the original graph. For clarity, we henceforth assume that $G$ has this $k$-partite form from the beginning, but we state our results as if it had not (e.g., we state that the blowup is $(O(nk),\tilde{O}(mk^2))$). 

We wish to add an expansion layer to $G$ without introducing any $k$-clique. Hence, we cannot na\"ively apply Algorithm \ref{sub:basic} on $G$ because, with some probability, some vertices from the expansion layer will participate in a $k$-clique. Instead, we add an expansion layer that is adjacent only to $V_1$, but with respect to the degrees of the vertices of $V_1$ in $G$. 
That is, we apply Algorithm \ref{sub:basic} in the following way: for every $v \in V_1$, $v$ will sample $\max\{\deg_G(v), C \log n \}$ edges to the expansion layer $U$. 
Denote the resulting graph by $G'$. See Figure \ref{fig:k-clique}.

\begin{figure}[h!]
    \centering
    \includegraphics[scale=0.8]{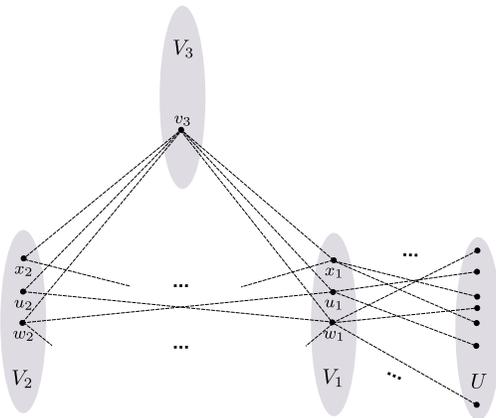}
    \caption{The construction of $k$-Clique for $k=3$. In this example, we can see that $v_3$ participates in the triangles $u_1,w_2,v_3$ and $w_1,u_2,v_3$.}
    \label{fig:k-clique}
\end{figure}

\begin{lemma}
\label{lem:k-clique}
With probability $1-O(\frac{1}{n})$, the conductance of $G'$ is $\phi(G') = \Omega(1/k^2)$.
\end{lemma}
\begin{proof}[Proof of Lemma \ref{lem:k-clique}]
We begin by introducing some notation. Let $S \subseteq V(G')$. Let $S_i := S \cap V_i$, $S_U := S \cap U$ and $S_V := S \cap V(G)$. We omit subscripts of $G'$. Assume w.l.o.g. that $|S_U| \leq |U|/2$ (otherwise take $\bar{S}$). To show that $\phi(S) = \Omega(1/k^2)$, we consider two cases:
\begin{enumerate}[label=(\arabic*)]
\item $\phi_{G}(S_V)$ is large. Then we show that $\phi(S)$ is also large with probability $1$.
\item $\phi_{G}(S_V)$ is small. Then we show that $vol_{G}(S_1)$ is at least some constant (that depends on $k$) fraction of $vol_{G}(S_V)$, and by Lemma \ref{lem:basic}, there are many cross edges between $S_1$ and $U$. Therefore $\phi(S)$ is large. 
\end{enumerate}
The first case is formalized by the next lemma.
\begin{lemma}
\label{lem:first-case}
For any cut $S \subseteq V(G')$, if $\phi_{G}(S_V) \geq p$ for some $0 < p <1$ then $\phi(S) \geq \frac{p}{3}$.
\end{lemma}
\begin{proof}
By the assumption that $\phi_{G}(S_V) \geq p$, we have:
\begin{equation*}
\begin{split}
&\phi(S) \geq \frac{\delta(S)}{vol(S)} = \frac{\delta_G(S_V) + \delta(S_U)}{vol(S)} \geq \frac{p \cdot vol_G(S_V) + e(V_1 \setminus S_1,S_U) + e(S_1,U \setminus S_U)}{vol(S)}.
\end{split}
\end{equation*}
Next, we upper bound $vol(S)$. Note that every edge counted by $vol(S)$ that does not belong to $G$ must be incident to either $S_1$ or $S_U$. We can thus write 
$
vol(S) \leq vol_{G}(S_V) + 2e(S_1,S_U) +  e(V_1 \setminus S_1,S_U) + e(S_1,U\setminus S_U) ,
$
and therefore:
\begin{equation*}
\begin{split}
&\frac{p \cdot vol_G(S_V) + e(V_1 \setminus S_1,S_U) + e(S_1,U \setminus S_U)}{vol(S)} \geq 
\frac{p \cdot vol_G(S_V) + e(V_1 \setminus S_1,S_U) + e(S_1,U \setminus S_U)}{vol_{G}(S_V) + 2e(S_1,S_U) + e(V_1 \setminus S_1,S_U) + e(S_1,U\setminus S_U)  } \geq \\
&\frac{p \cdot vol_G(S_V)}{vol_{G}(S_V) + 2e(S_1,S_U)}.
\end{split}
\end{equation*}
Since $e(S_1,S_U)$ consists of at most $vol_G(S_1)$ randomly sampled edges and at most $|S_1| \cdot C \log n$ deterministically added edges, we can say that $e(S_1,S_U) \leq 2vol_G(S_1) \leq vol_G(S_V)$. Hence we get:
\[
\frac{p \cdot vol_G(S_V)}{vol_{G}(S_V) + 2e(S_1,S_U)} \geq \frac{p \cdot vol_G(S_V)}{3 vol_G(S_V)} \geq \frac{p}{3}.
\]
\end{proof}
Thus, if $\phi_{G}(S_V) \geq \frac{0.1}{k^2}$ then by the above lemma we have $\phi(S) = \Omega(1/k^2)$ and we are done. Otherwise, we are in the second case, which is formalized by the next lemma:
\begin{lemma}
\label{lem:second-case}
If $\phi_{G}(S_V) < \frac{0.1}{k^2}$ then $vol_{G}(S_1) \geq \frac{0.3}{k^2} \cdot vol_{G}(S_V)$.
\end{lemma}

\begin{proof}
Let us define \emph{internal edges} in $S_V$ to be edges whose both endpoints are in $S_V$. It is easy to see that there are $({vol_{G}(S_V) - \delta_{G}(S_V)})/{2}$ internal edges because every such edge is counted twice. Now, since $\phi_{G}(S_V) < \frac{ 0.1}{k^2}$ then $\delta_{G}(S_V) < \frac{0.1}{k^2} \cdot vol_{G}(S_V)$. Therefore, there are at least:
\[
\frac{vol_{G}(S_V) - \delta_{G}(S_V)}{2} \geq \frac{(1-\frac{0.1}{k^2})vol_G(S_V)}{2} \geq 0.4 \cdot vol_G(S_V)
\]
internal edges in $S_V$. Since every internal edge belongs to one of the $\binom{k}{2}$ possible pairs of $S_1,S_2,\ldots,S_k$, there exists a ``heavy'' pair that contains at least:
\begin{equation}
\label{eq:heavy}
\frac{0.4 vol_{G}(S_V)}{ \binom{k}{2}} \geq \frac{0.4}{k^2} \cdot vol_{G}(S_V) 
\end{equation}
internal edges. If $S_1$ is incident to this pair, then we are done.
Otherwise, denote the heavy pair by $\{S_i,S_j\}$. Namely,
 $
 e(S_i,S_j) \geq 0.4 vol_{G}(S_V)/k^2.
 $
Since for every vertex $v \in S_i$, there is the same number of neighbors in $V_j$ and $V_1$, we have 
$
e(S_i,V_1) \geq 0.4 vol_G(S_V)/k^2.
$
Therefore:
\[
e(S_i,S_1) \geq e(S_i,V_i) - \delta_G(S_V) \geq  \frac{0.4}{k^2} \cdot vol_G(S_V) -  \frac{0.1}{k^2} \cdot vol_{G}(S_V) = \frac{0.3}{k^2} \cdot vol_G(S_V)
\]
\end{proof}
By Lemma \ref{lem:basic}, $S_1 \cup S_U$ must be an expanding cut in the subgraph $G_R := G'[S_1 \cup S_U]$, meaning that $\delta_{G_R}(S_1 \cup S_U) = \Omega(1)$. 
Note that by Lemma \ref{lem:second-case}, $vol_G(S_1)$ and therefore $vol_{G_R}(S_1)$ is a good represenative for $vol_G(S_V)$ (up to $1/k^2$).  
Hence, 
$
\delta(S) \geq \delta_{G_R}(S) = \Omega\left( vol_{G_R}(S_1) + vol_{G_R}(S_U) \right) = \Omega\left(\frac{vol(S_V) + vol(S_U)}{k^2} \right),
$
and therefore $\phi(S) = \Omega(1/k^2)$.
\end{proof}
\end{proof}

\subsection{Direct-WTER for Subgraph Isomorphism}
Fix $H$ to be a $k$-node pattern graph without pendant vertices.
\begin{proof}[Proof of Theorem \ref{thm:subgraph-isomorphism}]
Apply Algorithm \ref{sub:basic} on $G$ and then subdivide every $U$-to-$V(G)$ edge $k$ times. Denote by $G'$ be the resulting graph.
\begin{claim}
$H$ is a subgraph of $G$ if and only if it is a subgraph of $G'$.
\end{claim}
\begin{proof}
Since $H$ does not contain pendant vertices and $|V(H)| < k+1$, then any copy of $H$ in $G'$ does not use an edge that was added in the subdivision step. Therefore it is also a copy of $H$ in $G$.
\end{proof}
\begin{claim}
The expansion of $G'$ is $\phi(G') = \Omega(\frac{1}{k}) = \Omega(1)$.
\end{claim}
\begin{proof}
For any cut $S \subseteq V(G')$ denote by $\hat{S}$ the corresponding cut in $\hat{G}$, where $\hat{G}$ is the graph obtained from $G$ before the subdivision. More precisely, $\hat{S}$ consists of $S \cap V(G)$ and $S \cap U$.
By Theorem \ref{thm:basic}, we have $\phi_{\hat{G}}(S) = \Omega(1)$. Now, observe that $\delta(S) \geq \delta_{\hat{G}}(\hat{S})$ and that $vol(S) \leq 2k \cdot vol_{\hat{G}}(S) $. Therefore:
\[
\phi(S) \geq \frac{\delta_{\hat{G}}(\hat{S})}{2k \cdot \min\{ vol_{\hat{G}}(\hat{S}),vol_{\hat{G}}(\hat{S})\}} = \Omega\left(\frac{1}{k}\right)
\]
\end{proof}
\end{proof}

\subsection{ED-WTER for $4$-Cycle}
In this WTER we use the expander decomposition from Theorem \ref{thm:ed} to obtain a WTER for $4$-Cycle with conductance $n^{-\eps}$, for any constant $0<\eps < 0.5$, and running time $\tilde{O}(n^{2-0.5\eps} + n^{1.5+\eps})$. Observe that the number of outer edges in such decomposition is $\tilde{O}(n^{-\eps} m) = \tilde{O}(n^{1.5-\eps})$. We will call the vertices that are incident to outer edges \emph{portals}. 
We also assume that we have access to an oracle for $4$-Cycle, such that when given any $n^{-\eps}$-expander $X$ it answers ``YES'' if and only if $X$ contains a $4$-Cycle.
Note that if a single call to the oracle takes $O(n_i^{2-\delta})$ time for some $\delta>0$, where $n_i$ is the number of nodes in the $i^{th}$ expander and $\sum_i{n_i}=n$, then the total time for all calls will also be $\sum_{i}n_i^{2-\delta} =  O(n^{2-\delta})$; therefore, this WTER proves that any polynomial speedup on expanders implies a polynomial speedup in worst-case graphs.

\begin{proof}[Proof of Theorem \ref{thm:ed-wter} (ED-WTER for $4$-Cycle)]
Decomopse $G$ into expanders by using the expander decomposition of Theorem \ref{thm:ed} with parameter $\phi = n^{-\eps}$.
Next, query the oracle on every expander in the decomposition, and if at some point the oracle answers ``YES'' we are done.
Otherwise, we need to check if there is a $4$-Cycle that uses one of the $\tilde{O}(n^{1.5-\eps})$ outer edges. 

To this end, we use a ``high degree - low degree'' trick. First, we check whether there are $4$-cycles that contain any high-degree vertex in the graph (not necessarily a portal). We define a high-degree vertex to be a vertex of degree at least $n^{0.5(1+\eps)}$. Denote by $A \subseteq V(G)$ the set of all high-degree vertices in the graph. A simple counting argument shows that $ |A| = O( n^{1.5}/n^{0.5(1+\eps)}) = O(n^{1-0.5 \eps})$. Now, for every $v \in A$ we can check in $O(n)$ time whether $v$ participates in any $4$-Cycle in the following way: we mark every neighbor-of-a-neighbor of $v$, until we either find a collision between two marked nodes (that is, a $4$-Cycle) or until we marked all neighbors-of-a-neighbor. Note that in any case, we spend at most $O(n)$ time because after marking $n$ nodes, we must find a collision.  Hence, we can check in $O(|A| \cdot n ) = O(n^{2-0.5 \eps})$ time whether there is a vertex in $A$ that participates in a $4$-Cycle.

Next, we need to check if there is a $4$-Cycle that consists only of low-degree vertices and at least one outer edge. We denote by $P$ the set of all portals whose degree is less than
$n^{0.5(1+\eps)}$. For every vertex $v \in P$, we mark the endpoints of every $2$-path whose center is $v$, and that consists of at least one outer edge. If we denote by $b(v)$ the number of outer edges incident to $v$, then the number of such paths (and the time it takes to mark their endpoints) is at most :
\[
\sum_{v \in B} \deg(v) \cdot b(v) \leq n^{0.5(1+\eps)} \sum_{v \in B} b(v) = O(n^{0.5(1+\eps)} \cdot n^{1.5-\eps}) = O(n^{2-0.5\eps}).
\]
If we've encountered a collision between the two endpoints of different $2$-paths, we've found a $4$-Cycle, and we are done. Otherwise, we claim the graph is $4$-Cycle free. To see why, observe that any $4$-Cycle that contains an outer edge must contain at least two outer edges. Therefore, there is a pair of distinct portals $u,v$ that cover those two edges. Any such $4$-Cycle can be broken into two $2$-paths that pass through $u$ and $v$, each consisting of at least one outer edge. Therefore, this $4$-Cycle must be detected by the algorithm. 

The running time of this reduction is therefore $\tilde{O}( n^{1.5+\eps} + n^{2-0.5\eps})$.
\end{proof}

\subsection{Direct-WTER for Minimum Dominating Set}
Our Direct-WTER for Minimum Dominating Set (abbreviated as MDS) uses twins trick. That is, by adding a unique neighbor $u'$ to some vertex $u$, we can more easily control the size of the dominating set in the graph. Formally, we have:
\begin{fact}
\label{fact:domtwin}
If $u'$ is a twin of $u$, there exists a minimum dominating set that contains $u$.
\end{fact}
However, if we apply the same construction that worked for MCM (and MVC); adding an expansion layer $U$ and a twin to every $u \in U$, the vertices of $U$ will have to be in a minimum dominating set and therefore will also dominate $V(G)$. Essentially, the size of the minimum dominating set in $G'$ will become $MDS(G') := |U|$, losing all information about $MDS(G)$. To deal with this problem, we add an intermediate layer $R$ that consists of independent vertices that are ``copies'' of vertices in $G$. That is, for each $v_r \in R$ we have some $v \in V(G)$ and $v_r$ is adjacent in $G$ only to the neighborhood of $v$. Hence, $v_r$ doesn't dominate any vertices in $G$ that cannot be dominated by $v$. 

The vertices of $G$ that will have copies in $R$ are chosen by a randomized algorithm, such that every sparse cut in $G$ will have some copies in $R$, while keeping the size of $R$ less than $\eps n$. Then, we add an expansion layer $U$ to $R$, and to control the size of the dominating set we add twins to $U$. Note that $R$ becomes useless to any minimum dominating set because we can always replace copies with their originals. Let us now present the construction in full detail.
\begin{proof}[Proof of Theorem \ref{thm:mds} (Direct-WTER for Minimum Dominating Set)]
Let $0 < \eps$ be a constant and assume w.l.o.g. that $\eps$ is sufficiently small (say, $\eps \leq 0.01$). We construct the graph $G'$ using the following algorithm. The first four steps deal with sampling vertices from $G$ that hit all sparse cuts, and the last steps take copies of the sampled vertices and make the graph expander by adding an expansion layer.
\begin{enumerate}[label=(\arabic*)]
\item Let $Q = \emptyset$. For every $v \in V(G)$, add $v$ to $Q$ with probability $10\eps$, independently from other vertices. 
\item For every vertex $v \in V(G) \setminus Q$, we say that $v$ is \emph{bad} if $|N(v) \cap Q| < \eps \deg(v)$. Add to $Q$ all the bad vertices that have degree at least $\log (1/\eps)/\eps$.
\item Add vertices to $Q$ according to the next deterministic procedure. Decompose $G$ into connected, edge-disjoint subgraphs of size (where size is the number of vertices) in range $[1/\eps,3/\eps]$. From every subgraph in the decomposition pick one arbitrary vertex to $Q$. We can indeed get such decomposition using the next algorithm: 
\paragraph{Decomposition into connected, edge-disjoint subgraphs.} Take any rooted spanning tree of $G$. Let $v$ be the deepest node of degree greater than $2$. Every child of $v$ emanates a simple path to some leaf node. Divide such paths into components of size $1/\eps$ and a small remainder of size less than $1/\eps$ that is attached to $v$. Then, take an arbitrary union of remainders to obtain components of size in range $[1/\eps,2/\eps)$. After removing all the edges of the constructed components from the tree, the subtree of $v$ must be of size strictly less than $1/\eps$. At this stage we can treat $v$ as a vertex of degree at most $2$ with a remainder subtree of size less than $1/\eps$. We continue this process iteratively until the deepest vertex of degree greater than $2$ is the root. Then we might remain with a small connected component of size less than $1/\eps$. We can add this remainder to any other adjacent or incident component, so there will be one component whose size is at most $3/\eps$ instead of $2/\eps$.
\item Add to the graph a set of new vertices $R = \{v_r \mid v \in Q\}$, and for every $v_r \in R$ add all edges between $N_G(v)$ and $v_r$.
\item Add an expansion layer $U$ adjacent to $R$, with the following modifications to Algorithm \ref{sub:basic}. Denote by $\Delta_R := \max\{ \deg_G(v) \mid v_r \in R \}$.
\paragraph{Modifications to Algorithm \ref{sub:basic}.}
The modifications are similar to the ones we did in the constructions for MVC and $k$-Clique. For every $v_r \in R$, we wish to sample $deg_G(v) + C \log n$ edges from $v_r$ to $U$. 
Since we can not afford to add an expansion layer of size $|U| = 5(\Delta_R +  C\log n)$ (because $\Delta_R$ can be as large as $n$), we add an expansion layer of size $|U| = 5(\eps\Delta_R + C \log n)$. Then, every vertex $v_r \in R$ will sample $\eps \deg_G(v) + C \log n$ neighbors from $U$ (instead of $\deg_G(v) + C \log n$). Notice that we replaced the degree-increasing step by sampling $C \log n$ more neighbors from $U$. Thus, by Lemma \ref{lem:basic}, the induced subgraph on $R \cup U$ is an $\Omega(1)$-expander with high probability.
\item To control the size of the dominating set, for every $u \in U$ add a twin vertex $u'$.
\end{enumerate}
\begin{figure}[h!]
\centering
\includegraphics[scale=0.8]{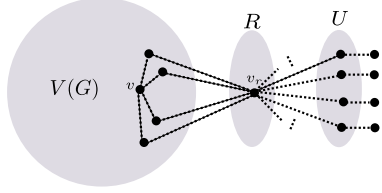}
\caption{The proposed WTER for dominating set. The vertices of $R$ are copies of vertices of $G$, thus sharing the same neighborhood in $G$.}
\label{fig:domset}
\end{figure}
Denote the resulting graph by $G'$. See Figure \ref{fig:domset}. 
Let us start by showing that the blowup is small. Formally, 
\begin{claim}
\label{claim:domset-blowup}
With high probability, $|V(G')| \leq n + 26\eps n$.
\end{claim}
\begin{proof}
Observe that the blowup in the number of vertices is $|R| + 2|U|$. Let us bound the size of $Q$ (and hence $R$). By the Chernoff bound, the number of sampled vertices in step (1) is at most $11\eps n$ with high probability. 
In addition, since for every $v \in V(G)$ we expect an $10 \eps$-fraction of its neighborhood to be hit by the sample, then by the Chernoff bound again, the probability that $v \in V(G)$ is bad is at most $\exp(-\eps \deg_G(v))$. 
Hence, the expected number of bad vertices whose degrees are at least $\log (1/\eps) / \eps$ (that we add to $Q$ in step (2)) is at most $n \cdot \exp(-\eps \cdot \log (1/\eps)/\eps) = \eps n$. 
Finally, the number of vertices added to $Q$ in step (3) is clearly at most $\eps n$. Overall, we have $|R| = |Q| \leq 14 \eps n$. We note that although the analysis of step (2) works in expectation, we can repeat this process to obtain the bound with high probability.

Next, let us bound $|U|$. By our modification to Algorithm \ref{sub:basic}, $|U| \leq 5(\eps \Delta_R + C \log n) \leq 6 \eps n$. Hence, the total blowup it at most $|R| + 2|U| \leq 14 \eps n + 12 \eps n = 26 \eps n$.
\end{proof}
\noindent
Therefore, to obtain a blowup of $\eps n$ we can apply the algorithm with $\eps' = \eps /26$. 
Now, let us show that computing $MDS(G)$ from $MDS(G')$ is trivial. Formally,
\begin{claim}
\label{claim:domset-size}
The size of the minimum dominating set in $G'$ is $MDS(G') = MDS(G) + |U|$.
\end{claim}
\begin{proof}
Note that by Fact \ref{fact:domtwin}, there exists a minimum dominating set $D$ (of $G'$) that contains $U$. In particular, $R$ is completely dominated by $U$ in this dominating set.
We can thus assume that $D$ does not contain any vertex from $R$ because we can always replace such vertices with their originals in $V(G)$. Namely, if $v_r \in D$, then we can remove it and add $v$ instead, and the resulting set will still be a dominating set of the same cardinality. Hence, since no vertex of $G$ is dominated by $U \cup R$, $D \cap V(G)$ must be a minimum dominating set for $G$. Hence, the cardinality of $D$ is $|D| = MDS(G) + |U|$.
\end{proof}
Now, let us prove that the conductance of $G'$ is $\Omega(1)$ (where the constant depends on $\eps$). Broadly speaking, we show that for any cut $S \subseteq V(G)$, if $S$ is not expanding inside $G$ (that is, $\delta_G(S) <<  vol_G(S)$) then there are many edges between $S$ and $R$, and since $R$ is adjacent to $U$ then we get the expansion we want, either between $S$ and $R$ or between $R$ and $U$. The first part of this claim is formalized by the next lemma:
\begin{lemma}
\label{lem:domset}
There exists a constant $c_\eps$ such that for every cut $S \subseteq V(G)$, either $\delta_G(S)  \geq c_\eps vol_G(S)$ or $e(S,R) \geq c_\eps vol_G(S)$.
\end{lemma}
\begin{proof}
Consider any cut $S \subseteq V(G)$. 
Observe that every vertex $v \in S$ that had been added to $Q$ (and hence $v_r \in R$) is good for us, because every neighbor of $v$ that is not in $S$ contributes an edge to $\delta_G(S)$, and every neighbor that is in $S$ has an edge to $v_r$ (and therefore contributes to $e(S,R)$). Thus, all the edges that are incident to $v$ have been taken care of.

Now, focus on the case where $v \in S$ is some vertex of degree at least $\log(1/\eps)/\eps$ that had not been added to $Q$. Then, at least $\eps$-fraction of its neighbors have been added to $Q$ in step (1), because otherwise, $v$ would have been considered ``bad'' and added to $Q$. Therefore, at least an $\eps$-fraction of $v$'s neighbors contribute their edges either to $\delta_G(S)$ or to $e(S,R)$, depending on whether they belong to $S$ or not, respectively.

Let us now deal with vertices of degree less than $\log (1/\eps)/\eps$ in $S$ and their incident edges. Denote by $H_1 , H_2 , \ldots, H_{k}$ the subgraphs of $G$ obtained from the decomposition in step (3). For every $i$ there is some $v^i \in V(H_i)$ that had been added to $Q$. Since $H_i$ is connected, then there is at least one edge incident to $v^i$ in $H_i$. 

Now, assume that $H_i$ contains low-degree vertices (that is, vertices of degree less than $\log(1/\eps)/\eps$) and that $S$ contains $V(H_i)$, for some $i \in [k]$. Then since $S$ also contains $v^i$ there is at least one edge between $V(H_i)$ and $R$. We charge all the edges incident to low-degree vertices in $H_i$ to this edge of $v^i_r$. Since there are at most $3/\eps$ vertices in $H_i$, we charged the edge by at most $3/\eps \cdot \log(1/\eps)/\eps = 3 \log(1/\eps)/\eps^2$ edges incident to low-degree vertices. If $S$ only intersects $V(H_i)$ but does not contain it, then since $H_i$ is connected there is some edge in $H_i$ that is in $\delta_G(S)$. The same charging argument works for this edge as well.

To summarize the proof, we showed that at least an $\eps$-fraction of the edges that are incident to high-degree vertices in $S$ are either present in $\delta_G(S)$ or in $e(S,R)$. We also showed that for every $3\log(1/\eps)/\eps^2$ edges that are incident to low-degree vertices in $S$, there is a unique edge either in $\delta_G(S)$ or in $e(S,R)$. Hence, let $c_\eps = \frac{1}{2} \cdot \eps^2/(3\log(1/\eps))$. Then either $\delta_G(S) \geq c_\eps vol_G(S)$ or $e(S,R) \geq c_\eps vol_G(S)$.
\end{proof}
Now we are ready to prove our main lemma. For any cut $S \subseteq V(G')$, we denote by $S_V,S_R,S_U$ the parts of $S$ in $V(G),R$ and $U$, respectively.
\begin{lemma}
\label{lem:domsetx}
With high probability, for any cut $S \subseteq V(G')$ we have $\phi(S) = \Omega(1)$.
\end{lemma}
\begin{proof}
Assume w.l.o.g. that $|S_U| \leq |U|/2$.  Observe that the number of internal edges in $S$ is at most some constant fraction of $vol_G(S_V) + vol(S_R)$ and therefore it suffices to show that $\delta(S) = \Omega(vol_G(S_V) + vol(S_R))$. Note that $\delta(S) \geq \delta_G(S_V) + e(S_V,R \setminus S_R) + e(S_R,U \setminus S_U)$. It follows from the proof of Lemma \ref{lem:basic} that $e(S_R,U \setminus S_U) = \Omega( vol(S_R) )$ with high probability and therefore it suffices to show that either $\delta_G(S_V) + e(S_V,R \setminus S_R) = \Omega(vol_G(S_V) )$ or that $vol(S_R) = \Omega vol_G(S_V)$. 

Therefore, if $\delta_G(S_V) = \Omega(vol_G(S_V))$ then we are done. 
Otherwise, by Lemma \ref{lem:domset} we have $e(S_V,R) = \Omega(vol_G(S_V))$.
If at least half those edges belong to $e(S_V,R \setminus S_R)$ then we are done again. 
Otherwise, we have $vol(S_R) = \Omega( vol_G(S_V) )$ and we are done. 
\end{proof}
\end{proof}

\section{Definition}
\label{sec:definition}

Let us conclude by formalizing the notion of a (fine-grained) WTER that captures Definitions \ref{def:direct} and \ref{def:ed-wter}.
\begin{definition}[$(\phi(n),a(n,m),b(n,m))$-WTER]
Problem $\mathcal{A}$ has a worst-case to expander-case self-reduction with conductance $\phi(n)$ and times $a(n,m),b(n,m)$ 
, if there exists a randomized algorithm with oracle access to $\mathcal{A}$, that solves $\mathcal{A}$ on any $n$-nodes, $m$-edges graph by making at most $k=k(n)$ calls to the oracle on instances $G_1,\ldots,G_{k}$, such that:
\begin{itemize}
    \item each $G_i$ is a $\phi(n)$-expander.
    \item for every $\eps>0$ there is a $\delta>0$ such that the reduction runs in time $O(a(n,m)^{1-\delta})$ and $\sum_{i=1}^{k} b(n_i,m_i)^{1-\eps} = O(a(n,m))^{1-\delta}$, where $n_i,m_i$ are the number of nodes and edges in $G_i$, respectively.
\end{itemize}
\end{definition}
\noindent 
We will abbreviate by saying that $\AG$ admits a $(\phi(n),a(n,m),b(n,m))$-WTER.
The intuition behind this definition can be described as follows. If a problem admits a $(\phi(n),a(n,m),b(n,m))$-WTER, then any ``fast'' (i.e. $b(n,m)^{1-\eps}$ time) algorithm on $\phi(n)$-expanders can be used to obtain a ``faster'' ($a(n,m)^{1-\delta}$ time) algorithm on worst-case graphs. For example, if a problem admits a $((\log n)^{-1},n^3,n^2)$-WTER. Then the existence of any $O(n^{2(1-\eps)})$-time algorithm that solves the problem on $n$-vertex $(\log n)^{-1}$-expanders implies that the problem can be solved on worst-case graphs in time $O(n^{3(1-\delta)})$. In this work we mainly focused on $(\Omega(1),a(n,m),a(n,m))$-WTERs, where $a(n,m)$ is the best known upper bound to solve the problem, essentialy showing the problem is as easy on worst-case graphs as on expanders.

For completeness, we formalize the intutition in the next theorem:
\begin{theorem}
If problem $\mathcal{A}$ has a $(\phi(n),a(n,m),b(n,m))$-WTER, then for all $\eps>0$ there is a $\delta>0$ such that if there is an algorithm that solves $\mathcal{A}$ on $\phi(n)$-expanders in time $O(b(n,m)^{1-\eps})$, then $\mathcal{A}$ can be solved on worst-case graphs in time $O(a(n,m)^{1-\delta})$.
\end{theorem}

\begin{proof}
Suppose that problem $\mathcal{A}$ has a $(\phi(n),a(n,m),b(n,m))$-WTER and that there is an algorithm that runs in time $O(b(n,m)^{1-\eps})$ that solves the problem on $\phi(n)$-expanders. Then we can simulate the algorithm in the WTER on worst-case graphs by replacing the oracle with the fast algorithm on expanders. The total running time is $O(a(n,m)^{1-\delta})$, for some $\delta>0$.
\end{proof}

\section*{Acknowledgments}
We thank the anonymous reviewers for their helpful comments that helped us improve the paper.

\bibliographystyle{plainurl}
\bibliography{references}

\begin{thebibliography}{10}

\bibitem{ABW18}
Amir Abboud, Arturs Backurs, and Virginia~Vassilevska Williams.
\newblock If the current clique algorithms are optimal, so is {V}aliant's
  parser.
\newblock {\em {SIAM} J. Comput.}, 47(6):2527--2555, 2018.
\newblock \href {https://doi.org/10.1137/16M1061771}
  {\path{doi:10.1137/16M1061771}}.

\bibitem{ABKZ22}
Amir Abboud, Karl Bringmann, Seri Khoury, and Or~Zamir.
\newblock Hardness of approximation in {P} via short cycle removal: cycle
  detection, distance oracles, and beyond.
\newblock In Stefano Leonardi and Anupam Gupta, editors, {\em {STOC} '22: 54th
  Annual {ACM} {SIGACT} Symposium on Theory of Computing, Rome, Italy, June 20
  - 24, 2022}, pages 1487--1500. {ACM}, 2022.
\newblock \href {https://doi.org/10.1145/3519935.3520066}
  {\path{doi:10.1145/3519935.3520066}}.

\bibitem{ACK20}
Amir Abboud, Vincent Cohen{-}Addad, and Philip~N. Klein.
\newblock New hardness results for planar graph problems in {P} and an
  algorithm for sparsest cut.
\newblock In Konstantin Makarychev, Yury Makarychev, Madhur Tulsiani, Gautam
  Kamath, and Julia Chuzhoy, editors, {\em Proccedings of the 52nd Annual {ACM}
  {SIGACT} Symposium on Theory of Computing, {STOC} 2020, Chicago, IL, USA,
  June 22-26, 2020}, pages 996--1009. {ACM}, 2020.
\newblock \href {https://doi.org/10.1145/3357713.3384310}
  {\path{doi:10.1145/3357713.3384310}}.

\bibitem{AD16}
Amir Abboud and S{\o}ren Dahlgaard.
\newblock Popular conjectures as a barrier for dynamic planar graph algorithms.
\newblock In Irit Dinur, editor, {\em {IEEE} 57th Annual Symposium on
  Foundations of Computer Science, {FOCS} 2016, 9-11 October 2016, Hyatt
  Regency, New Brunswick, New Jersey, {USA}}, pages 477--486. {IEEE} Computer
  Society, 2016.
\newblock \href {https://doi.org/10.1109/FOCS.2016.58}
  {\path{doi:10.1109/FOCS.2016.58}}.

\bibitem{AGV15}
Amir Abboud, Fabrizio Grandoni, and Virginia~Vassilevska Williams.
\newblock Subcubic equivalences between graph centrality problems, apsp and
  diameter.
\newblock In {\em Proceedings of the twenty-sixth annual ACM-SIAM symposium on
  Discrete algorithms}, pages 1681--1697. SIAM, 2014.

\bibitem{AKLPST21}
Amir Abboud, Robert Krauthgamer, Jason Li, Debmalya Panigrahi, Thatchaphol
  Saranurak, and Ohad Trabelsi.
\newblock Gomory-hu tree in subcubic time.
\newblock {\em arXiv preprint arXiv:2111.04958}, 2021.

\bibitem{AKT21}
Amir Abboud, Robert Krauthgamer, and Ohad Trabelsi.
\newblock Subcubic algorithms for gomory-hu tree in unweighted graphs.
\newblock In Samir Khuller and Virginia~Vassilevska Williams, editors, {\em
  {STOC} '21: 53rd Annual {ACM} {SIGACT} Symposium on Theory of Computing,
  Virtual Event, Italy, June 21-25, 2021}, pages 1725--1737. {ACM}, 2021.
\newblock \href {https://doi.org/10.1145/3406325.3451073}
  {\path{doi:10.1145/3406325.3451073}}.

\bibitem{AKT22friendly}
Amir Abboud, Robert Krauthgamer, and Ohad Trabelsi.
\newblock Friendly cut sparsifiers and faster {G}omory-{H}u trees.
\newblock In {\em Proceedings of the 2022 Annual ACM-SIAM Symposium on Discrete
  Algorithms (SODA)}, pages 3630--3649. SIAM, 2022.

\bibitem{DBLP:conf/focs/AbboudW14}
Amir Abboud and Virginia~Vassilevska Williams.
\newblock Popular conjectures imply strong lower bounds for dynamic problems.
\newblock In {\em 55th {IEEE} Annual Symposium on Foundations of Computer
  Science, {FOCS} 2014, Philadelphia, PA, USA, October 18-21, 2014}, pages
  434--443. {IEEE} Computer Society, 2014.
\newblock \href {https://doi.org/10.1109/FOCS.2014.53}
  {\path{doi:10.1109/FOCS.2014.53}}.

\bibitem{AVW16}
Amir Abboud, Virginia~Vassilevska Williams, and Joshua~R. Wang.
\newblock Approximation and fixed parameter subquadratic algorithms for radius
  and diameter in sparse graphs.
\newblock In Robert Krauthgamer, editor, {\em Proceedings of the Twenty-Seventh
  Annual {ACM-SIAM} Symposium on Discrete Algorithms, {SODA} 2016, Arlington,
  VA, USA, January 10-12, 2016}, pages 377--391. {SIAM}, 2016.
\newblock \href {https://doi.org/10.1137/1.9781611974331.ch28}
  {\path{doi:10.1137/1.9781611974331.ch28}}.

\bibitem{AR18}
Udit Agarwal and Vijaya Ramachandran.
\newblock Fine-grained complexity for sparse graphs.
\newblock In {\em Proceedings of the 50th Annual ACM SIGACT Symposium on Theory
  of Computing}, pages 239--252, 2018.

\bibitem{Kaplan22}
Daniel Agassy, Dani Dorfman, and Haim Kaplan.
\newblock Expander decomposition with fewer inter-cluster edges using a
  spectral cut player.
\newblock {\em CoRR}, abs/2205.10301, 2022.
\newblock \href {http://arxiv.org/abs/2205.10301} {\path{arXiv:2205.10301}},
  \href {https://doi.org/10.48550/arXiv.2205.10301}
  {\path{doi:10.48550/arXiv.2205.10301}}.

\bibitem{alev2017graph}
Vedat~Levi Alev, Nima Anari, Lap~Chi Lau, and Shayan~Oveis Gharan.
\newblock Graph clustering using effective resistance.
\newblock {\em arXiv preprint arXiv:1711.06530}, 2017.

\bibitem{AlmanW20}
Josh Alman and Virginia~Vassilevska Williams.
\newblock A refined laser method and faster matrix multiplication.
\newblock In {\em Proceedings of the 2021 {ACM-SIAM} Symposium on Discrete
  Algorithms, {SODA} 2021}, pages 522--539, 2021.
\newblock \href {https://doi.org/10.1137/1.9781611976465.32}
  {\path{doi:10.1137/1.9781611976465.32}}.

\bibitem{alon1995color}
Noga Alon, Raphael Yuster, and Uri Zwick.
\newblock Color-coding.
\newblock {\em Journal of the ACM (JACM)}, 42(4):844--856, 1995.

\bibitem{AYZ97}
Noga Alon, Raphael Yuster, and Uri Zwick.
\newblock Finding and counting given length cycles.
\newblock {\em Algorithmica}, 17(3):209--223, 1997.
\newblock \href {https://doi.org/10.1007/BF02523189}
  {\path{doi:10.1007/BF02523189}}.

\bibitem{asadi2022worst}
Vahid~R Asadi, Alexander Golovnev, Tom Gur, and Igor Shinkar.
\newblock Worst-case to average-case reductions via additive combinatorics.
\newblock {\em arXiv preprint arXiv:2202.08996}, 2022.

\bibitem{BRSV17}
Marshall Ball, Alon Rosen, Manuel Sabin, and Prashant~Nalini Vasudevan.
\newblock Average-case fine-grained hardness.
\newblock In Hamed Hatami, Pierre McKenzie, and Valerie King, editors, {\em
  Proceedings of the 49th Annual {ACM} {SIGACT} Symposium on Theory of
  Computing, {STOC} 2017, Montreal, QC, Canada, June 19-23, 2017}, pages
  483--496. {ACM}, 2017.
\newblock \href {https://doi.org/10.1145/3055399.3055466}
  {\path{doi:10.1145/3055399.3055466}}.

\bibitem{BW12}
Nikhil Bansal and Ryan Williams.
\newblock Regularity lemmas and combinatorial algorithms.
\newblock {\em Theory Comput.}, 8(1):69--94, 2012.
\newblock \href {https://doi.org/10.4086/toc.2012.v008a004}
  {\path{doi:10.4086/toc.2012.v008a004}}.

\bibitem{bast2006matching}
Holger Bast, Kurt Mehlhorn, Guido Schafer, and Hisao Tamaki.
\newblock Matching algorithms are fast in sparse random graphs.
\newblock {\em Theory of Computing Systems}, 39(1):3--14, 2006.

\bibitem{beigel1999finding}
Richard Beigel.
\newblock Finding maximum independent sets in sparse and general graphs.
\newblock In {\em SODA}, volume~99, pages 856--857. Citeseer, 1999.

\bibitem{BEPR12}
Nicolas Bourgeois, Bruno Escoffier, Vangelis~T Paschos, and Johan~MM van Rooij.
\newblock Fast algorithms for max independent set.
\newblock {\em Algorithmica}, 62(1):382--415, 2012.

\bibitem{CPS21}
Yi{-}Jun Chang, Seth Pettie, Thatchaphol Saranurak, and Hengjie Zhang.
\newblock Near-optimal distributed triangle enumeration via expander
  decompositions.
\newblock {\em J. {ACM}}, 68(3):21:1--21:36, 2021.
\newblock \href {https://doi.org/10.1145/3446330} {\path{doi:10.1145/3446330}}.

\bibitem{chen2006improved}
Jianer Chen, Iyad~A Kanj, and Ge~Xia.
\newblock Improved parameterized upper bounds for vertex cover.
\newblock In {\em International symposium on mathematical foundations of
  computer science}, pages 238--249. Springer, 2006.

\bibitem{MFlinear22}
Li~Chen, Rasmus Kyng, Yang~P. Liu, Richard Peng, Maximilian~Probst Gutenberg,
  and Sushant Sachdeva.
\newblock Maximum flow and minimum-cost flow in almost-linear time.
\newblock {\em CoRR}, abs/2203.00671, 2022.
\newblock \href {http://arxiv.org/abs/2203.00671} {\path{arXiv:2203.00671}},
  \href {https://doi.org/10.48550/arXiv.2203.00671}
  {\path{doi:10.48550/arXiv.2203.00671}}.

\bibitem{chuzhoy2019deterministic}
Julia Chuzhoy, Yu~Gao, Jason Li, Danupon Nanongkai, Richard Peng, and
  Thatchaphol Saranurak.
\newblock A deterministic algorithm for balanced cut with applications to
  dynamic connectivity, flows, and beyond.
\newblock In {\em 61st {IEEE} Annual Symposium on Foundations of Computer
  Science, {FOCS} 2020}, pages 1158--1167, 2020.
\newblock \href {https://doi.org/10.1109/FOCS46700.2020.00111}
  {\path{doi:10.1109/FOCS46700.2020.00111}}.

\bibitem{CK19}
Julia Chuzhoy and Sanjeev Khanna.
\newblock A new algorithm for decremental single-source shortest paths with
  applications to vertex-capacitated flow and cut problems.
\newblock In Moses Charikar and Edith Cohen, editors, {\em Proceedings of the
  51st Annual {ACM} {SIGACT} Symposium on Theory of Computing, {STOC} 2019,
  Phoenix, AZ, USA, June 23-26, 2019}, pages 389--400. {ACM}, 2019.
\newblock \href {https://doi.org/10.1145/3313276.3316320}
  {\path{doi:10.1145/3313276.3316320}}.

\bibitem{Cygan+16}
Marek Cygan, Holger Dell, Daniel Lokshtanov, D{\'{a}}niel Marx, Jesper
  Nederlof, Yoshio Okamoto, Ramamohan Paturi, Saket Saurabh, and Magnus
  Wahlstr{\"{o}}m.
\newblock On problems as hard as {CNF-SAT}.
\newblock {\em {ACM} Trans. Algorithms}, 12(3):41:1--41:24, 2016.
\newblock \href {https://doi.org/10.1145/2925416} {\path{doi:10.1145/2925416}}.

\bibitem{DLW20}
Mina Dalirrooyfard, Andrea Lincoln, and Virginia~Vassilevska Williams.
\newblock New techniques for proving fine-grained average-case hardness.
\newblock In Sandy Irani, editor, {\em 61st {IEEE} Annual Symposium on
  Foundations of Computer Science, {FOCS} 2020, Durham, NC, USA, November
  16-19, 2020}, pages 774--785. {IEEE}, 2020.
\newblock \href {https://doi.org/10.1109/FOCS46700.2020.00077}
  {\path{doi:10.1109/FOCS46700.2020.00077}}.

\bibitem{EvaldD16}
Jacob Evald and S{\o}ren Dahlgaard.
\newblock Tight hardness results for distance and centrality problems in
  constant degree graphs.
\newblock {\em CoRR}, abs/1609.08403, 2016.
\newblock URL: \url{http://arxiv.org/abs/1609.08403}, \href
  {http://arxiv.org/abs/1609.08403} {\path{arXiv:1609.08403}}.

\bibitem{FGK06}
Fedor~V Fomin, Fabrizio Grandoni, and Dieter Kratsch.
\newblock Measure and conquer: A simple ${O}(2^{0.288n})$ independent set
  algorithm.
\newblock In {\em SODA}, volume~6, pages 18--25. Citeseer, 2006.

\bibitem{fomin2004exact}
Fedor~V Fomin, Dieter Kratsch, and Gerhard~J Woeginger.
\newblock Exact (exponential) algorithms for the dominating set problem.
\newblock In {\em International Workshop on Graph-Theoretic Concepts in
  Computer Science}, pages 245--256. Springer, 2004.

\bibitem{oded_expander}
Oded Goldreich.
\newblock Basic facts about expander graphs.
\newblock In {\em Studies in Complexity and Cryptography. Miscellanea on the
  Interplay between Randomness and Computation}, pages 451--464. Springer,
  2011.

\bibitem{HLS22}
Monika Henzinger, Andrea Lincoln, and Barna Saha.
\newblock The complexity of average-case dynamic subgraph counting.
\newblock In Joseph~(Seffi) Naor and Niv Buchbinder, editors, {\em Proceedings
  of the 2022 {ACM-SIAM} Symposium on Discrete Algorithms, {SODA} 2022, Virtual
  Conference / Alexandria, VA, USA, January 9 - 12, 2022}, pages 459--498.
  {SIAM}, 2022.
\newblock \href {https://doi.org/10.1137/1.9781611977073.23}
  {\path{doi:10.1137/1.9781611977073.23}}.

\bibitem{henzinger2022fine}
Monika Henzinger, Ami Paz, and AR~Sricharan.
\newblock Fine-grained complexity lower bounds for families of dynamic graphs.
\newblock {\em arXiv preprint arXiv:2208.07572}, 2022.

\bibitem{KVV04}
Ravi Kannan, Santosh Vempala, and Adrian Vetta.
\newblock On clusterings: Good, bad and spectral.
\newblock {\em Journal of the ACM}, 51(3):497--515, 2004.
\newblock \href {https://doi.org/10.1145/990308.990313}
  {\path{doi:10.1145/990308.990313}}.

\bibitem{KLOS14}
Jonathan~A Kelner, Yin~Tat Lee, Lorenzo Orecchia, and Aaron Sidford.
\newblock An almost-linear-time algorithm for approximate max flow in
  undirected graphs, and its multicommodity generalizations.
\newblock In {\em Proceedings of the 25th Annual ACM-SIAM Symposium on Discrete
  Algorithms}, pages 217--226. SIAM, 2014.

\bibitem{KM22}
Alexander~S. Kulikov and Ivan Mihajlin.
\newblock Polynomial formulations as a barrier for reduction-based hardness
  proofs.
\newblock {\em CoRR}, abs/2205.07709, 2022.
\newblock \href {http://arxiv.org/abs/2205.07709} {\path{arXiv:2205.07709}},
  \href {https://doi.org/10.48550/arXiv.2205.07709}
  {\path{doi:10.48550/arXiv.2205.07709}}.

\bibitem{LS21}
Jason Li and Thatchaphol Saranurak.
\newblock Deterministic weighted expander decomposition in almost-linear time.
\newblock {\em CoRR}, abs/2106.01567, 2021.
\newblock URL: \url{https://arxiv.org/abs/2106.01567}, \href
  {http://arxiv.org/abs/2106.01567} {\path{arXiv:2106.01567}}.

\bibitem{LWW18}
Andrea Lincoln, Virginia~Vassilevska Williams, and R.~Ryan Williams.
\newblock Tight hardness for shortest cycles and paths in sparse graphs.
\newblock In Artur Czumaj, editor, {\em Proceedings of the Twenty-Ninth Annual
  {ACM-SIAM} Symposium on Discrete Algorithms, {SODA} 2018, New Orleans, LA,
  USA, January 7-10, 2018}, pages 1236--1252. {SIAM}, 2018.
\newblock \href {https://doi.org/10.1137/1.9781611975031.80}
  {\path{doi:10.1137/1.9781611975031.80}}.

\bibitem{YS22}
Yaowei Long and Thatchaphol Saranurak.
\newblock Near-optimal deterministic vertex-failure connectivity oracles.
\newblock {\em CoRR}, abs/2205.03930, 2022.
\newblock \href {http://arxiv.org/abs/2205.03930} {\path{arXiv:2205.03930}},
  \href {https://doi.org/10.48550/arXiv.2205.03930}
  {\path{doi:10.48550/arXiv.2205.03930}}.

\bibitem{MV80}
Silvio Micali and Vijay~V. Vazirani.
\newblock An ${O}(\sqrt{|v|} \cdot |{E}|)$ algorithm for finding maximum
  matching in general graphs.
\newblock In {\em 21st Annual Symposium on Foundations of Computer Science,
  Syracuse, New York, USA, 13-15 October 1980}, pages 17--27. {IEEE} Computer
  Society, 1980.
\newblock \href {https://doi.org/10.1109/SFCS.1980.12}
  {\path{doi:10.1109/SFCS.1980.12}}.

\bibitem{motwani1994average}
Rajeev Motwani.
\newblock Average-case analysis of algorithms for matchings and related
  problems.
\newblock {\em Journal of the ACM (JACM)}, 41(6):1329--1356, 1994.

\bibitem{MS04}
Marcin Mucha and Piotr Sankowski.
\newblock Maximum matchings via gaussian elimination.
\newblock In {\em 45th Annual IEEE Symposium on Foundations of Computer
  Science}, pages 248--255. IEEE, 2004.

\bibitem{NSW17}
Danupon Nanongkai, Thatchaphol Saranurak, and Christian {Wulff-Nilsen}.
\newblock Dynamic minimum spanning forest with subpolynomial worst-case update
  time.
\newblock In {\em 2017 IEEE 58th Annual Symposium on Foundations of Computer
  Science (FOCS)}, pages 950--961. IEEE, 2017.

\bibitem{OSV12}
Lorenzo Orecchia, Sushant Sachdeva, and Nisheeth~K. Vishnoi.
\newblock Approximating the exponential, the {L}anczos method and an
  {$\tilde{O}(m)$}-time spectral algorithm for balanced separator.
\newblock In {\em Proceedings of the 44th Annual ACM Symposium on Theory of
  Computing}, pages 1141--1160, 2012.
\newblock \href {https://doi.org/10.1145/2213977.2214080}
  {\path{doi:10.1145/2213977.2214080}}.

\bibitem{OV11}
Lorenzo Orecchia and Nisheeth~K Vishnoi.
\newblock Towards an {SDP}-based approach to spectral methods: A
  nearly-linear-time algorithm for graph partitioning and decomposition.
\newblock In {\em Proceedings of the 22nd Annual ACM-SIAM Symposium on Discrete
  Algorithms}, pages 532--545. SIAM, 2011.

\bibitem{razgon2009faster}
Igor Razgon.
\newblock Faster computation of maximum independent set and parameterized
  vertex cover for graphs with maximum degree 3.
\newblock {\em Journal of Discrete Algorithms}, 7(2):191--212, 2009.

\bibitem{SW19}
Thatchaphol Saranurak and Di~Wang.
\newblock Expander decomposition and pruning: Faster, stronger, and simpler.
\newblock In {\em Proceedings of the 30th Annual {ACM-SIAM} Symposium on
  Discrete Algorithms, {SODA} 2019}, pages 2616--2635, 2019.
\newblock \href {https://doi.org/10.1137/1.9781611975482.162}
  {\path{doi:10.1137/1.9781611975482.162}}.

\bibitem{ST13}
Daniel~A Spielman and Shang-Hua Teng.
\newblock A local clustering algorithm for massive graphs and its application
  to nearly linear time graph partitioning.
\newblock {\em SIAM Journal on computing}, 42(1):1--26, 2013.

\bibitem{ST14}
Daniel~A Spielman and Shang-Hua Teng.
\newblock Nearly linear time algorithms for preconditioning and solving
  symmetric, diagonally dominant linear systems.
\newblock {\em SIAM Journal on Matrix Analysis and Applications},
  35(3):835--885, 2014.

\bibitem{van2011exact}
Johan~MM Van~Rooij and Hans~L Bodlaender.
\newblock Exact algorithms for dominating set.
\newblock {\em Discrete Applied Mathematics}, 159(17):2147--2164, 2011.

\bibitem{williams2018some}
Virginia~Vassilevska Williams.
\newblock On some fine-grained questions in algorithms and complexity.
\newblock In {\em Proceedings of the International Congress of Mathematicians:
  Rio de Janeiro 2018}, pages 3447--3487. World Scientific, 2018.

\bibitem{WW18}
Virginia~Vassilevska Williams and R.~Ryan Williams.
\newblock Subcubic equivalences between path, matrix, and triangle problems.
\newblock {\em J. {ACM}}, 65(5):27:1--27:38, 2018.
\newblock \href {https://doi.org/10.1145/3186893} {\path{doi:10.1145/3186893}}.

\bibitem{woeginger2003exact}
Gerhard~J Woeginger.
\newblock Exact algorithms for {NP}-hard problems: A survey.
\newblock In {\em Combinatorial optimization—eureka, you shrink!}, pages
  185--207. Springer, 2003.

\bibitem{YZ97}
Raphael Yuster and Uri Zwick.
\newblock Finding even cycles even faster.
\newblock {\em {SIAM} J. Discret. Math.}, 10(2):209--222, 1997.
\newblock The conference version appeared in ICALP 1994.
\newblock \href {https://doi.org/10.1137/S0895480194274133}
  {\path{doi:10.1137/S0895480194274133}}.

\end{thebibliography}
\end{document}